\documentclass{article}
\usepackage{arxiv}

\usepackage[utf8]{inputenc} % allow utf-8 input
\usepackage[T1]{fontenc}    % use 8-bit T1 fonts
\usepackage{hyperref}       % hyperlinks
\usepackage{url}            % simple URL typesetting
\usepackage{booktabs}       % professional-quality tables
\usepackage{amsfonts}       % blackboard math symbols
\usepackage{nicefrac}       % compact symbols for 1/2, etc.
\usepackage{microtype}      % microtypography
\usepackage{lipsum}		% Can be removed after putting your text content
\usepackage{graphicx}
\usepackage{natbib}
\usepackage{doi}
\usepackage{amsmath}

\usepackage{amssymb,amsfonts,amsthm,bm,bbm}
\usepackage{algorithm}
\usepackage{soul}
\usepackage{algorithmic}
\usepackage{array}
\usepackage{graphicx}
\usepackage{booktabs}
\usepackage{multirow}
\usepackage{pifont}
\usepackage{enumitem}
\usepackage{multicol}
\usepackage{amsmath,amssymb,amsfonts,amsthm}
\usepackage{graphicx}
\usepackage{comment}
\usepackage{diagbox}
\usepackage{bigfoot} % to allow verbatim in footnote
\usepackage{thmtools}
\usepackage{thm-restate}
\usepackage{soul}

\usepackage{algorithm,algorithmic}
\usepackage{subcaption}
\usepackage{url}
\usepackage{amsmath}

\usepackage{amssymb,amsfonts,amsthm,bm,bbm}
\usepackage{algorithmic}
\usepackage{algorithm}
\usepackage{array}
\usepackage{graphicx}
\usepackage{booktabs}
\usepackage{multirow}
\usepackage{pifont}
\usepackage{hyperref}

\usepackage[american]{babel}
\usepackage{amsfonts}
\usepackage{amsthm}
\usepackage{multirow}
\usepackage{color, colortbl}
\usepackage{enumitem}
\usepackage{amsfonts}
\usepackage{multirow}
\usepackage{booktabs}
\usepackage{amsthm}
\usepackage{xcolor}
\usepackage{tikz} % nice language for creating drawings and diagrams
\usepackage{hyperref}
\usepackage{caption}
\usepackage{subcaption}
\usepackage{multicol}

    \usepackage[many]{tcolorbox}    
\newtcolorbox{mybox}[1]{%
    tikznode boxed title,
    enhanced,
    arc=0mm,
    interior style={white},
    attach boxed title to top center= {yshift=-\tcboxedtitleheight/2},
    fonttitle=\bfseries,
    colbacktitle=white,coltitle=black,
    boxed title style={size=normal,colframe=white,boxrule=0pt},
    title={#1}}

\usepackage{enumitem}
\usepackage{multicol}
\usepackage{xr} 

\makeatletter
\newcommand*{\addFileDependency}[1]{% argument=file name and extension
  \typeout{(#1)}
  \@addtofilelist{#1}
  \IfFileExists{#1}{}{\typeout{No file #1.}}
}
\makeatother

\usepackage{xcolor}
\usepackage[printonlyused,withpage]{acronym}
\usepackage[section]{placeins}

\newtheorem{theorem}{\textbf{Theorem}}

\newtheorem{claim}{Claim}
\newtheorem{lemma}[theorem]{\textbf{Lemma}}
\theoremstyle{definition}
\newtheorem{definition}{Definition}
\newcommand{\proname}[0]{$\textsf{pRFT}$}

\title{Towards Rational Consensus in Honest Majority}

\date{May 13, 2024}	% Here you can change the date presented in the paper title
\author{
    \href{https://orcid.org/0000-0002-5662-0386}
    {\includegraphics[scale=0.06]{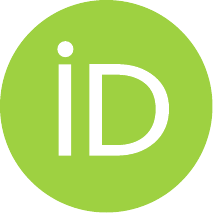}\hspace{1mm}Varul Srivastava} \\
	Machine Learning Lab\\
	IIIT Hyderabad\\
	\texttt{varul.srivastava@research.iiit.ac.in} \\
	\And
	\href{https://orcid.org/0000-0003-4634-7862}{\includegraphics[scale=0.06]{orcid.pdf}\hspace{1mm}Sujit Gujar} \\
	Machine Learning Lab\\
	IIIT Hyderabad\\
	\texttt{sujit.gujar@iiit.ac.in} \\
}

\renewcommand{\shorttitle}{Rational Consensus in Honest Majority}

\hypersetup{
pdftitle={Rational Consensus in Honest Majority},
pdfsubject={q-bio.NC, q-bio.QM},
pdfauthor={Varul Srivastava, Sujit Gujar},
pdfkeywords={Consensus, Rational Agreement, Blockchains, Distributed Systems},
}

\begin{document}
\maketitle
% \renewcommand{\shortauthors}{Srivastava and Gujar}

\begin{abstract}
Distributed consensus protocols reach agreement among $n$ players in the presence of $f$ adversaries; different protocols support different values of $f$. Existing works study this problem for different adversary types (captured by threat models). There are three primary threat models: (i) Crash fault tolerance (CFT), (ii) Byzantine fault tolerance (BFT), and (iii) Rational fault tolerance (RFT), each more general than the previous. Agreement in repeated rounds on both (1) the proposed value in each round and (2) the ordering among agreed-upon values across multiple rounds is called Atomic BroadCast (ABC). ABC is more generalized than consensus and is employed in blockchains.

This work studies ABC under the RFT threat model. We consider $t$ byzantine and $k$ rational adversaries among $n$ players. 
We also study different types of rational players based on their utility towards (1) liveness attack, (2) censorship or (3) disagreement (forking attack). We study the problem of ABC under this general threat model in partially-synchronous networks. 
We show (1) ABC is impossible for $n/3< (t+k) <n/2$ if rational players prefer liveness or censorship attacks and (2) the consensus protocol proposed by Ranchal-Pedrosa and Gramoli cannot be generalized to solve ABC due to insecure Nash equilibrium (resulting in disagreement). For ABC in partially synchronous network settings, we propose a novel protocol \textsf{pRFT}(practical Rational Fault Tolerance). We show \textsf{pRFT} achieves ABC if (a) rational players prefer only disagreement attacks and (b) $t < \frac{n}{4}$ and $(t + k) < \frac{n}{2}$. In \textsf{pRFT}, we incorporate accountability (capturing deviating players) within the protocol by leveraging honest players. We also show that the message complexity of \textsf{pRFT} is at par with the best consensus protocols that guarantee accountability.
\end{abstract}

\keywords{distributed consensus, blockchains, security, equilibrium analysis}



%%%%%%%%%%%%%%%%%%%%%%%%%%%%%%%%%%%%%%%%%
\section{Introduction}
\label{sec:introduction}
%%%%%%%%%%%%%%%%%%%%%%%%%%%%%%%%%%%%%%%%%

\emph{Agreement and Distributed Consensus} is a well-studied problem since its introduction~\cite{Pease1980LSPInception,Lamport1982BGP} as the Byzantine Generals' Problem. Applications include maintaining distributed file systems, building fault-tolerant systems, and, most recently, in Blockchain technology. Consensus is reaching agreement on a common value $v$ among a set of players $n$ with $f$ faulty players. In case of repeated consensus, we require an additional condition that ordering among the agreed-upon values is the same across rounds. This generalization of agreement is called \emph{Atomic BroadCast} (ABC). 

Prior works achieved consensus in the presence of up to $t < n/3$ Byzantine failures~\cite{Lamport1982BGP,CastroLiskov1999pBFT} under synchronous network assumptions.
\citet{CastroLiskov1999pBFT} extended consensus to partially synchronous network through pBFT protocol. FLP Impossibility~\cite{Fischer1985ImpossibilityAsync} stated that agreement through a \emph{deterministic} protocol in asynchronous settings is impossible in the presence of even one faulty player. Later, randomized protocols were proposed~\cite{Cachin2000AsyncAgreement,Cachin2001AsyncBroadcast,Bracha1987AsyncBFT}, which achieve consensus in asynchronous network settings for $t < n/3$. Different consensus protocols work under different threat models. For instance, Paxos~\cite{Lamport1998Paxos,Lamport2001Paxos} and Raft~\cite{Ongaro2014Raft} achieve consensus in presence of $c$ \emph{crash fault} players (where players can go offline). This threat model is Crash Fault Tolerant -- $CFT(c)$. pBFT~\cite{CastroLiskov1999pBFT}, and Honeybadger~\cite{Miller2016Honeybadger} achieve consensus in the presence of $t<\frac{n}{3}$ byzantine faults (where the player can follow any arbitrary strategy). This is the Byzantine Fault Tolerant threat model -- $BFT(t)$. 

Ranchal-Pedrosa and Gramoli~\cite{Gramoli2022Trap} introduced a general rational threat model where $t$ byzantine players and $k$ rational players can collude. This is called a rational threat model -- $RFT(k,t)$ and the agreement problem called Rational Consensus (RC). The authors propose RC using \emph{baiting based protocol} -- TRAP by showing the existence of a Nash Equilibrium (NE) that achieves consensus for $t < n/3$ and $k + t < n/2$. Protocols are called \emph{Nash Incentive Compatible} (NIC) when the honest strategy is NE. However, we show the existence of another (more preferred) NE strategy that causes disagreement for TRAP when used to solve Atomic Broadcast (ABC). The rational players may prefer this dystopic equilibrium point over the more improbable secure equilibrium, making the protocol insecure. Game-theoretic security under the existence of multiple Nash equilibrium points is realized when following the protocol is Pareto-optimal/Focal equilibrium~\cite{Schelling1963SoC} \footnote{for details, refer to discussion in Section~\ref{ssec:eq-ic} or \cite{Schelling1963SoC}}. Protocols ensure stronger security guarantees if the equilibrium is Dominant strategy equilibrium (DSE) instead of Nash equilibrium (NE). 

There is an absence of protocols realizing ABC in the rational threat model. Our work addresses this gap and proves impossibilities and a novel protocol \proname\ that achieves ABC in the rational threat model under certain conditions on rational players' utility.

\subsection*{Our Approach} 
This work generalizes the $RFT(k,t)$ model~\cite{Gramoli2022Trap}: (i) to incorporate payoff in repeated rounds with discounting and (ii) to model rational different agent types. These rational player types are depicted by $\theta$: (1) $\theta = 1$ is incentivized towards disagreement, (2) $\theta = 2$ is incentivized towards censorship attack, and (3) $\theta = 3$ is incentivized towards denial of service (liveness is compromised). We show that for rational players types $\theta=2,3$, achieving Rational Consensus (RC) is not possible under the $RFT(t,k)$ threat model with $\frac{n}{3} \leq t + k < \frac{n}{2}$. 
Hence, we focus on rational players of type $\theta=1$. We propose \proname, which achieves consensus in $RFT(t,k)$ threat model for $\frac{n}{3} \leq t + k < \frac{n}{2}$ when rational players are of type\footnote{for details on different types of rational players, refer to Section~\ref{sssec:players}} $\theta=1$. Previous protocols attempted to incentive engineer the protocol such that rational players are incentivized to ``bait'' the deviating collusion. The rational players are incentivized to bait. The baiting is an equilibrium for the rational players if a certain threshold number ($m$) of players bait, leading to the protocol's security. However, this method was susceptible to the existence of alternate, insecure equilibrium points.

We capture deviation without relying on rational players (by incorporating accountability within the protocol), guaranteeing the capture of deviating players, making them prone to penalty. In our proposal, each player deposits some collateral. If deviation is captured, deviating players lose the collateral, which a rational player does not want. Hence, it will deviate if its collateral is intact. It leads to following the protocol as a dominant strategy for all rational players. Thus, our protocol achieves a stronger game-theoretic security guarantee, namely, Dominant Strategy Incentive Compatibility (DSIC). Based on these results, we place our work (coloured blue) in Table~\ref{tab:performance-thresholds} among other known bounds for consensus in different network settings and threat models. We also show that \proname\ achieves message complexity and message size, at par with the best available message complexity amongst protocols that guarantee \emph{accountability} such as~\cite{Ranchal2020ZLBAccountable,Civit2021Polygraph} and \proname\ works under a more general threat model than these protocols. 

\definecolor{Gray}{gray}{0.9}
\begin{table}[!t]
\centering
\begin{small}
\begin{tabular}{ p{10em} | p{10em}  p{10em} p{12em}}
\toprule
\multirow{2}{*}{\textbf{Network}} & \multicolumn{3}{c}{\textbf{Threat Model}}\\
 & $CFT(c)$ & $BFT(t)$ & $RFT(t,k)$ \\
\midrule
\rowcolor{Gray}
Synchronous & $2c < n$~\cite{Lamport1998Paxos} & $2t < n$~\cite{Pease1980LSPInception} & $t < \frac{n}{2}, k < \frac{n}{2}$~\cite{Pease1980LSPInception}\\
Partially-synchronous & $2c < n$~\cite{Lamport1998Paxos} & $3t < n$~\cite{CastroLiskov1999pBFT}  & \color{blue}$t < \frac{n}{4}, t + k < \frac{n}{2}$\\
\rowcolor{Gray}
Asynchronous & $c < \frac{n}{3}$~\cite{Bracha1987AsyncBFT} & $t < \frac{n}{3}$~\cite{Bracha1987AsyncBFT}  & $t < \frac{n}{3}$~\cite{Bracha1987AsyncBFT}\\ 
\bottomrule
\end{tabular}\\
The results highlighted in \textcolor{blue}{blue} are contributions of our work.
\smallskip
\caption{Bounds for consensus in different threat models.}
\label{tab:performance-thresholds}
\end{small}
\end{table}
%%%%%%%%%%%%%%%%%%%%%%%%%%%%%%%%%%%%%%%%%
\subsection{Our Contributions}
\label{ssec:our-contributions}
%%%%%%%%%%%%%%%%%%%%%%%%%%%%%%%%%%%%%%%%%
In this paper, we first extend the Rational Threat Model proposed in~\cite{Gramoli2022Trap}. We classify the rational players into three types represented by different values of $\theta$ (representing player types). $\theta=3$ is when rational players are incentivized to compromise liveness and cause censorship or disagreement. $\theta=2$ is when rational players are incentivized only to cause censorship or disagreement, and $\theta=1$ is when players are incentivized only to cause disagreement. Based on this for $k$ rational, $t$ byzantine players such that $t < t_{0}$ and total players are $n$, we present the following impossibilities in Section~\ref{sec:impossibilities}.
\begin{itemize}[leftmargin=8pt]
    \item[-] consensus is not possible when the set of rational players are of type $\theta=3$ for $k + t < n/2$ and $t_{0} < n/3$ in partially synchronous and asynchronous settings (Theorem~\ref{thm:r3-impossible}).
    \item[-] consensus is not possible when the set of rational players are of type $\theta=2$ for $k + t < n/2$ and $t < n/3$ in partially synchronous and asynchronous settings (Theorem~\ref{thm:r2-impossible}).
    \item[-] There exists an additional Nash equilibrium that results in disagreement in baiting-based consensus protocols (such as TRAP, proposed in~\cite{Gramoli2022Trap}) (Theorem~\ref{thm:trap-impossible}). Thus, there is a need for a new agreement protocol with one equilibrium point (preferably guaranteeing a stronger notion of Dominant Strategy Equilibrium instead of Nash Equilibrium\footnote{for details regarding notions of equilibrium points, see Section~\ref{ssec:eq-ic}}).
\end{itemize}

Following this, we propose a novel protocol \proname\ (Section~\ref{sec:our-protocol}) which achieves consensus in $RFT(k,t)$ threat model. We show the following results for \proname.
\begin{itemize}[leftmargin=8pt]
    \item[-] \proname\ achieves consensus with $k + t < n/2$ and $t < n/4$ when rational players are of type $\theta = 1$. 
    \item[-] \proname\ guarantees correctness (Dominant Strategy Equilibrium) and liveness in synchronous and partially synchronous network settings. 
    \item[-] We show that \proname\ achieves optimal message complexity among consensus protocols that provide accountability\footnote{see Section~\ref{sssec:proof-of-fraud} for a formal definition of accountable consensus protocols}. 
\end{itemize}

%%%%%%%%%%%%%%%%%%%%%%%%%%%%%%%%%%%%%%%%%
\section{Related Work}
\label{sec:related-work}
%%%%%%%%%%%%%%%%%%%%%%%%%%%%%%%%%%%%%%%%%
The domain of distributed consensus has had extensive research in the past $50$ years. We discuss below the work which is closely related to our work.

\paragraph{Byzantine Agreement.} The inception of Byzantine Consensus saw the formulation of protocols under synchronous network settings~\cite{Pease1980LSPInception,Lamport1982BGP}. This foundational work was subsequently extended to encompass partially synchronous scenarios~\cite{Cynthia1988PartialSynchrony}. In the context of an asynchronous network, Fischer et al.~\cite{Fischer1985ImpossibilityAsync} introduced the FLP impossibility by which it is impossible to reach an agreement using a \emph{deterministic} protocol in the asynchronous network in the presence of even one faulty process. To overcome this, randomized protocols for distributed agreement and broadcast~\cite{Cachin2000AsyncAgreement,Cachin2001AsyncBroadcast} in asynchronous settings were introduced. Blockchain technology introduced through Nakamoto's seminal whitepaper~\cite{nakamoto2009Bitcoin} solves the State Machine Replication (SMR) using alternate protocols such as Proof-of-Work (PoW) and Proof-of-Stake (PoS) for consensus in public settings and BFT based Atomic Broadcast/Agreement (ABA) in the permissioned setting.

\paragraph{Atomic BroadCast (ABA)} Protocols achieving ABA were introduced through Paxos~\cite{Lamport1998Paxos,Lamport2001Paxos} and Raft~\cite{Ongaro2014Raft} which assumed a conservative \emph{crash-fault} threat model and synchronous network assumptions. Deterministic protocols as pBFT~\cite{CastroLiskov1999pBFT}, Hotstuff~\cite{Abraham2020SyncHotstuff,Yin2019Hotstuff}, FlexibleBFT~\cite{Malkhi2019FlexibleBFT} and others~\cite{Civit2021Polygraph,Aublin2013RBFT} achieved ABA in byzantine threat model. However, these protocols worked in synchronous and partially synchronous networks. Randomized ABA protocols such as Honeybadger~\cite{Miller2016Honeybadger}, Tardigrade~\cite{blum2021tardigrade} and others~\cite{Gao2022Async2,Lu2022Async4,Micali2017Algorand,Spiegelman2022Async3,zhang2022Async1} achieve agreement even in asynchronous settings. The threat model used by these protocols is the byzantine threat model, which contains $t$ byzantine adversaries. While in synchronous network agreement protocols~\cite{blum2021tardigrade,Abraham2020SyncHotstuff} tolerate $n > 2f$, in partially-synchronous and asynchronous network agreement protocols~\cite{blum2021tardigrade,Miller2016Honeybadger} tolerate $n > 3f$. 

\paragraph{Rational Agreement} The Byzantine threat model used to analyze distributed cryptographic protocols was extended by adding rational players. Aiyer et al.~\cite{Aiyer2005BAR} introduced the BAR (Byzantine Altruistic Rational) framework, where players/processes are altruistic, byzantine and rational. Rational players deviate only if utility from deviating is more than following the protocol honestly. BAR model has been since used to solve the distributed cryptographic problem called Transfer Problem~\cite{Vilacca2011BARIntro,Fernandez2011BARPODS}, which they solve when the producer of $N$ processors is such that $N > 2f$ for $f$ (byzantine) faulty processes. Other works~\cite{Abraham2006RationalMPC,Asharov2011RationalCrypto,Badertscher2021BitcoinRPD51,Garay2013RPD,Ganesh2022RationalCrypto} analyze the security of distributed cryptographic and game theoretic protocols against rational threat model. Some distributed protocols implement trusted third-party mediators using \emph{cheap talks}. This process of realizing cheap talks secure against $k$ rational and $t$ byzantine adversaries for $k + 2t < n$ in synchronous settings~\cite{Abraham2006RationalMPC} and $3(k + t) < n$ in asynchronous settings~\cite{Abraham2019AsyncCheapTalks} similar to consensus protocols under byzantine adversaries~\cite{Katz2006thresholds,blum2021tardigrade}. \citet{katz2012RC} proposed analysis of agreement in the presence of two types of players --- honest and rational adversaries. \citet{Gramoli2022Trap} proposed TRAP protocol, which achieves rational consensus in the partially-synchronous and asynchronous network when $2(k + t) < n$ and $3t < n$. However, their result on the sufficiency of TRAP in achieving RC relies on rational players opting for an optimistic (but less plausible) equilibrium point instead of a dystopic (but more realistic) equilibrium, as demonstrated in this work. To our knowledge, the threat model discussed in~\cite{Gramoli2022Trap} is the most general present in literature, where $t$ byzantine and $k$ rational players exist, allowing arbitrary collusion among them. Our work extends this model by generalizing the type of rational player and presenting various exciting results in this setting. 

%%%%%%%%%%%%%%%%%%%%%%%%%%%%%%%%%%%%%%%%%
\section{Preliminaries}
\label{sec:prelims}
%%%%%%%%%%%%%%%%%%%%%%%%%%%%%%%%%%%%%%%%%

In this section, we motivate some definitions and prior works which are relevant to our work. 

%%%%%%%%%%%%%%%%%%%%%%%%%%%%%%%%%%%%%%%%%
\subsection{Blockchain and Agreement}
\label{ssec:blockchains}
%%%%%%%%%%%%%%%%%%%%%%%%%%%%%%%%%%%%%%%%%
Blockchain technology achieves \emph{ Atomic Broadcast} (ABC) (formally defined in Section~\ref{ssec:prelims-consensus}), i.e. repeated consensus while preserving the total ordering of agreed-upon values. The probabilistic agreement is reached in permissionless settings using protocols such as PoW~\cite{nakamoto2009Bitcoin} and PoS~\cite{Micali2017Algorand}. Permissioned blockchains use BFT type of consensus where a committee (comprised of a set of players $\mathcal{P} = \{\mathcal{P}_{1},\mathcal{P}_{2},\ldots,\mathcal{P}_{n}\}$) 
proposes and achieves agreement on a single value in each round. In the context of Blockchain, the agreed-upon value is a Block. The block comprises \emph{state-changes} to the global state of the system in the form of transactions. Each block $B$ has a set of transactions $\overline{tx}$ and points to the parent block, i.e. the block agreed upon immediately before it. 

\smallskip \noindent The result is a chain of agreed-upon blocks, represented by $C$. Each player $\mathcal{P}_{i}$ has their own version of this chain, represented by $C_{i}$. 
\cite{CastroLiskov1999pBFT,Miller2016Honeybadger} wait for all nodes to agree on the same value. Other solutions like~\cite{Micali2017Algorand} and \proname\ (our solution) partially confirm the blocks subject to rollbacks in case of adversarial behaviour.  These blocks are labelled \emph{Tentative Blocks}. Such blocks might be rolled back once the network synchronizes. They are considered finalized only if followed by a finalized block (defined in~\cite{Micali2017Algorand} as a block mined during a phase of synchrony). If the last \emph{final block} was mined $z$ blocks before the most recent block, then \emph{common-prefix property}~\cite{Garay2015Backbone} holds if $\cap_{i=1}^{n} C^{\lfloor z}_{i}$ (i.e. chain obtained by removing the $z$ most recent blocks in each player $\mathcal{P}_{i}$'s chain) is a prefix of all $C_{i}$.

\paragraph{Player Types} We briefly discuss the type of players that are a part of our discussion. The system consists of $n$ players, out of which $h$ are honest, $t$ are byzantine and $k$ are rational. 
\begin{itemize}[leftmargin=8pt]
    \item \textbf{Honest Players:} Also called altruistic players, they follow the specified protocol honestly. 
    \item \textbf{Byzantine Players:} They follow any strategy with the intent to cause disruption in the correctness, liveness or other properties guaranteed by the distributed protocol. They are immune to incentive manipulation and will choose a strategy irrespective of their \emph{payoff} from that strategy.
    \item \textbf{Rational Players:} These players follow the strategy which gives the highest \emph{payoff} based off of some utility structure (which is protocol and agent type specific\footnote{for more details on utility structure in our case, refer to Section~\ref{sssec:utility}}). Therefore, such players deviate from following the protocol honestly only if there exists a strategy with a higher payoff.  
\end{itemize}
%%%%%%%%%%%%%%%%%%%%%%%%%%%%%%%%%%%%%%%%%
\subsection{Flavours of Consensus}
\label{ssec:prelims-consensus}
%%%%%%%%%%%%%%%%%%%%%%%%%%%%%%%%%%%%%%%%%
The problem of consensus in a distributed setting was first motivated by the Byzantine General's problem~\cite{Lamport1982BGP} and has since been discussed in the literature in different capacities such as Byzantine Broadcast (BB), Byzantine Agreement~\cite{Pease1980LSPInception,Lamport1982BGP} (BA), Atomic Broadcast~\cite{Cristian1995ATOMICBF,blum2021tardigrade} (ABC) and Rational Consensus~\cite{Aiyer2005BAR,Gramoli2022Trap} (RC). In this section, we will discuss notions of consensus that serve as preliminaries to our work. We elaborate some of the more common definitions in Appendix~\ref{app:definitions}.

\subsubsection{Byzantine Broadcast \& Agreement}
Byzantine Generals' Problem was introduced by Lamport et al.~\cite{Pease1980LSPInception}. 
In the Byzantine Agreement problem, the system is comprised of $t$ faulty (byzantine) and $n - t$ non-faulty (honest/altruistic) players. We motivate the definition of BA from~\cite[Definition 2]{blum2021tardigrade}.

\subsubsection{Atomic Broadcast}
\emph{Atomic Broadcast }(ABC) is a generalization of BA. In BA, players reach an agreement on a value $v$ once, while in ABC, players reach this agreement multiple times, maintaining a ledger of agreed-upon values. ABC is therefore repeated rounds of agreement on values such that a ledger is maintained with an added constraint that the ordering of different values is the same for all honest players. The formal definition of ABC is motivated from~\cite[Definition 5]{blum2021tardigrade}.

\subsubsection{Rational Consensus}

With the increasing interest in the rational security analysis of distributed cryptographic protocols, we motivate from~\cite{Gramoli2022Trap} and define \textit{Rational Consensus} (RC) -- the equivalent of ABC with a general (byzantine and rational) threat model as follows. We motivate the definition of robustness from~\cite{Gramoli2022Trap} and extend it to repeated rounds by adding the condition of $c$\textsf{-strict ordering} (the rational equivalent of ABC). 
\begin{definition}[$(t,k)$\textsf{-robustness}]\label{def:rational-consensus}
    Consider a protocol $\Pi$ is run by $\mathcal{P} = \{\mathcal{P}_{1},\mathcal{P}_{2}\ldots,\mathcal{P}_{n}\}$ players where $t$ players are byzantine and $k$ players are rational 
    %\sg{not sure if we defined rational players b4}
    (follow the strategy with the highest incentive) while remaining $n - t - k$ players are honest. The protocol $\Pi$ is $(t,k)$\textsf{-robust} if it satisfies:
    \begin{itemize}[leftmargin=8pt]
        \item[-] $(t,k)$\textsf{-validity} If all altruistic players receive value $v$ then they all agree on value $v$
        \item[-] $(t,k)$\textsf{-agreement} All altruistic players output the same value in each round. 
        \item[-] $c-$\textsf{strict ordering} If the ledger of agreed blocks is $C_{1}$ and $C_{2}$ for two altruistic players with $|C_{1}| \leq |C_{2}|$, then $C_{1}^{\lfloor c} \subseteq C_{2}^{\lfloor c}$
        holds\footnote{$C^{\lfloor c}$ is the ledger after removing the last $c$ blocks}.
        \item[-] $(t,k)$\textsf{-eventual liveness} if an honest player outputs $B$ then all honest players output $B$ eventually.
    \end{itemize}
\end{definition}

\noindent We define a stronger notion of RC when protocols also satisfy censorship resistance (we define $(t,k)$\textsf{-censorship resistance}) and call such RC protocols as \textsf{strongly} $(t,k)$\textsf{-robust}. 
\begin{definition}[$(t,k)$\textsf{-censorship resistance}]
    A protocol $\Pi$ satisfies $(t,k)$\textsf{-censorship resistance} if when all honest players have transaction $tx$ as input, then eventually all honest players output a \textsf{block} with transaction $tx$.
\end{definition}

\begin{definition}[\textsf{strong} $(t,k)$\textsf{-robustness}]
    A protocol $\Pi$ is \textsf{strongly} $(t,k)$\textsf{-robust} if $\Pi$ is $(t,k)$\textsf{-robust} and $(t,k)$\textsf{-censorship resistant}. 
\end{definition}

%%%%%%%%%%%%%%%%%%%%%%%%%%%%%%%%%%%%%%%%%
\subsection{Cryptographic and Network Preliminaries}
\label{ssec:prelims-crypto}
%%%%%%%%%%%%%%%%%%%%%%%%%%%%%%%%%%%%%%%%%
\paragraph{Digital Signatures} We employ the use of digital signatures and assume unforgeability except with negligible probability by all players (players are polynomially bounded) with access to random oracle\footnote{such players represent the set of all Probabilistic Polynomial Time Turing Machines (PPTM)}. This use of PKI (Public Key Infrastructure) for unforgeable digital signature has been employed for Authenticated Byzantine Agreement, first introduced by Dolev \& Strong~\cite{Dolev1983AuthenticatedByzantine}.

\paragraph{Trusted Setup} Before initiation of the protocol, we assume there is a trusted broadcast-type setup similar to~\cite{Cachin2000AsyncAgreement} (implemented via a common third party) where all participating players share their public keys, against which any digitally signed message is verified. 

\paragraph{Network Settings} We assume reliable channels between each pair of players involved in our analysis. Therefore, messages cannot be lost or tampered with, but they can face network delays. Based on the dealy, we consider three types of networks: (1) \emph{synchronized} is when the delay is upper bounded by a known bound $\Delta$ which can be used to parameterize the protocol. (2) \emph{asynchronous} network does not have an upper bound on the delay, but the message eventually gets delivered (i.e. delay for each message is finite). (3) \emph{partially-synchronous}~\cite{Cynthia1988PartialSynchrony} network is when the system behaves as an asynchronous network till before an event called \emph{Global Stabilization Time} (GST), after which the system becomes synchronous with some upper bound on delay.

\subsection{Baiting based Consensus Protocols}
\label{ssec:trap}

Baiting-based consensus protocols such as~\cite{Gramoli2022Trap} assume collusion of $k + t$ players ($k$ rational and $t$ byzantine) deviating from the agreement protocol. They employ a baiting strategy to incentivize rational players to bait the collusion by submitting Proof-of-Fraud, which consists of $t_{0} + 1$ conflicting signatures (for a detailed discussion on proof-of-fraud see Section~\ref{sssec:proof-of-fraud} and~\cite{Civit2021Polygraph,Ranchal2020ZLBAccountable}). If $m$ rational players follow the baiting strategy, then one of them is randomly selected for the reward $\mathcal{R}$ associated with baiting. Each player has a deposit $\mathcal{L}$, which they lose if there is Proof-of-Fraud containing their conflicting signatures. Additionally, there is a net utility gain of $\mathcal{G}$ for the collusion $K \cup T$ if the system ends up in disagreement. This utility is equally divided between the set of rational players such that each player $P_{i} \in K$ gets $\frac{\mathcal{G}}{k}$ payoff. For a more extensive description of the system and results used in TRAP, refer to~\cite{Gramoli2022Trap}.

%%%%%%%%%%%%%%%%%%%%%%%%%%%%%%%%%%%%%%%%%
\section{Our Model}
\label{sec:model}
%%%%%%%%%%%%%%%%%%%%%%%%%%%%%%%%%%%%%%%%%
We model RC as a game between three types of players Byzantine, Rational and Altruistic (Honest) Players. In this section, we define (i) the game, (ii) the utility structure and  (iii) network models.  
%%%%%%%%%%%%%%%%%%%%%%%%%%%%%%%%%%%%%%%%%
\subsection{The Game}
\label{ssec:game}
%%%%%%%%%%%%%%%%%%%%%%%%%%%%%%%%%%%%%%%%%
The Game consists of a set of players $\mathcal{P} = \{\mathcal{P}_{1},\mathcal{P}_{2}\ldots,\mathcal{P}_{n}\}$ who maintain a ledger of \textsf{Blocks}. Each \textsf{Block} contains a set of transactions $\overline{tx} = (tx_{i})_{i=1}^{z}$ which are valid wrt. previously \emph{confirmed} blocks. Agreement on a single \textsf{Block} proceeds in discrete intervals called \emph{rounds}. In each round $r$ we have a leader $\mathcal{P}_{l}$ (for $l = 1 + (r\;\text{mod}\;n)$) that proposes a block $B_{r}$. 

\subsubsection{Players}
\label{sssec:players}
We have three types of players: Byzantine, Rational and Honest players.
\begin{itemize}[leftmargin=8pt]
    \item[-] \textbf{Byzantine:} There are $t$ byzantine players belonging to set $T \subset \mathcal{P}$. They follow any arbitrary strategy irrespective of payoff, with the goal of causing maximum disruption of the system. 
    \item[-] \textbf{Rational:} There are $k$ rational players belonging to the set $K \subset \mathcal{P}$ which follow the strategy that provides them maximum utility. They follow the protocol $\Pi$ honestly unless there exists a deviation that gives them more than \emph{negligible} advantage in utility. The rational players can be of one of four types, represented by $\theta \in \{0,1,2,3\}$.
    \item[-] \textbf{Honest:} There are $h = n - k - t$ Honest players belonging to the set $H \subseteq \mathcal{P}$. These players follow the protocol honestly as long as participation in the protocol is incentivized over abstaining from participation, otherwise, these players don't participate in the protocol. 
\end{itemize}
To allow a larger attack space for the adversary, we consider that the Rational and Byzantine players can collude with each other. Therefore, there can exist a collusion set $ \subseteq K \cup T$ of size $\leq k + t$. 

\paragraph{System States} Due to strategies followed by the players in the system and due to the external environment (network delays), a distributed system can be in the following states:

\begin{itemize}[leftmargin=8pt]
    \item \textbf{No Progress ($\sigma_{NP}$):} In any round $r (\forall r \in \mathbb{R})$ no new blocks are mined. 
    \item \textbf{Conditional Progress ($\sigma_{CP}$):} In any round $r (\forall r \in \mathbb{R})$ the confirmed blocks contain transactions such that $\forall tx_{i} \in \overline{tx}, tx_{i} \not\in Z$ where $Z$ is the set of \emph{censored transactions}.
    \item \textbf{Disagreement ($\sigma_{Fork}$):} In any round $r (\forall r \in \mathbb{R})$ we have two honest players $\mathcal{P}_{i},\mathcal{P}_{j} \in H$ such that their ledger state has two \emph{confirmed blocks} $B_{i}$ and $B_{j}$ at the same height $h$ and $B_{i} \neq B_{j}$.
    \item \textbf{Honest Execution ($\sigma_{0}$):} In any round, the protocol executes according to the honest execution and does not violate correctness or liveness conditions. 
\end{itemize} 

\begin{table}[!t]
\begin{small}
\centering
\begin{tabular}{ p{7em} | p{3em} p{3em} p{3em} p{3em} | p{12em}}
\toprule
\multirow{2}{*}{\textbf{Player Type ($\theta$)}} & \multicolumn{4}{c}{\textbf{System State ($\sigma$)}} & \multirow{2}{*}{\textbf{Preferred States}}\\
 & $\sigma_{NP}$ & $\sigma_{CP}$ & $\sigma_{Fork}$ & $\sigma_{0}$ & \\
\midrule
$\theta = 3$ & $\alpha$ & $\alpha$ & $\alpha$ & $0$ & No Progress, Censorship, Fork\\
$\theta = 2$ & $-\alpha$ & $\alpha$ & $\alpha$  & $0$ & Censorship, Fork\\
$\theta = 1$ & $-\alpha$ & $-\alpha$ & $\alpha$  & $0$ & Fork\\ 
$\theta = 0$ & $-\alpha$ & $-\alpha$ & $-\alpha$  & $0$ & Honest Execution\\
\bottomrule
\end{tabular}
\smallskip
\caption{Payoff function $f(\sigma,\theta)$}
\label{tab:func1}
\end{small}
\end{table}

\paragraph{Player Types} We further model rational player type\footnote{The notion of player type $\theta$ corresponds only to rational players because Byzantine type ($\theta = 3$) and Altruistic/Honest type ($\theta = 0$) is already fixed by definition.} through $\theta \in \{0,1,2,3\}$. The players' types depend on their incentives for different states of the distributed system. We characterize the function $f(\sigma,\theta)$ for payoff when player type is $\theta$ and the system is in state $\sigma$. This function is represented in Table~\ref{tab:func1} (for some positive constant $\alpha$). Types of rational players are described below:

\begin{itemize}[leftmargin=8pt]
    \item \textsf{$\theta = 3$:} Such players are incentivized if the system state is $\sigma_{NP}, \sigma_{CP} \text{ or } \sigma_{Fork}$. % That is, $f(\sigma,3) = \alpha$ for $\sigma \in \{\sigma_{CP},\sigma_{NP},\sigma_{Fork}\}$.
    \item \textsf{$\theta = 2$:} Such players are incentivized if system state is $\sigma_{CP} \text{ or } \sigma_{Fork}$.
    \item \textsf{$\theta = 1$:} Such players are incentivized if system state is $\sigma_{Fork}$.
    \item $\theta = 0$ This type of rational player is disincentivized if the system is in any state except $\sigma_{0}$. However, they might be incentivized against sending messages or performing verification of messages, which has been previously analysed by~\citet{Amoussou2020AAMASRCEquilibrium}. 
\end{itemize}
If there are multiple types of rational players, we analyze for security for the worst types amongst them. If $K_{i}$ is set of rational players with type $\theta = i$, we say, $K = \cup_{i=0}^{3}K_{i}$ is of type $\theta = \max \{i | K_{i} \neq \emptyset\}$. 

\subsubsection{Utility Structure}
\label{sssec:utility}
Byzantine players follow the strategy that causes maximal disruption to the protocol $\Pi$ irrespective of the associated utility. If the protocol is \emph{individually rational}, honest players follow $\Pi$ honestly. Therefore, we need only model the utility for the rational players $\mathcal{P}_{i} \in K$. 

\paragraph{Strategy Space} In addition to defining the possible states of the system and the possible types of a rational player, we define the set of strategies available with the rational players in each round. 
\begin{itemize}[leftmargin=8pt]
    \item \textbf{Abstain $\pi_{abs}$:} $\mathcal{P}_{i}$ does not send messages in the particular phase or round.
    \item \textbf{Double-Sign $\pi_{ds}$:} $\mathcal{P}_{i}$ signs on two conflicting messages in the same phase of the same round. 
    \item \textbf{Honest $\pi_{0}$:} $\mathcal{P}_{i}$ follows the specified protocol $\Pi$ honestly. 
\end{itemize}

\paragraph{Penalty} There also exists a penalty that is incurred by player $\mathcal{P}_{i}$ if there exists proof that with overwhelming probability, the player has deviated from the protocol. The penalty is a fixed constant $L$ for each player which is the collateral deposited by these players before participating in consensus. If proof of malicious behaviour is found, this deposit is stashed/burnt~\cite{Kiayias2020PoB} and the penalty mechanism should be such that for a player that has followed the protocol honestly the penalty should be $0$ except with negligible probability.\footnote{This condition ensures \emph{Individual Rationality} of Honest players} The penalty is determined through a function $D(\pi,\sigma)$ based on the mechanism which takes value $1$ if a penalty is incurred and $0$ otherwise.
 
Based on the strategy $\pi$ and type $\theta$, we define the utility of rational players in a round $r$ by taking an expectation over the set of possible states, i.e. $\sigma \in S$ as:  $u_{i}(\pi,\theta,r) =  \mathbb{E}_{\sigma\sim\;S}[f(\sigma,\theta)] - L\cdot\;D(\pi,\sigma)$

\if 0
\begin{equation}\label{eqn:utility}
    \begin{aligned}
    u_{i}(\pi,\theta,r) = & \mathbb{E}_{\sigma\sim\;S}[f(\sigma,\theta)] - L\cdot\;D(\pi,\sigma) 
    \end{aligned}
\end{equation} \fi

The function $f:S\times\{0,1,2,3\}\rightarrow\{-\alpha,0,\alpha\}$ (for some positive constant $\alpha$) is given in Table~\ref{tab:func1}. If we consider the expected utility in a particular round for player $\mathcal{P}_{i} \in K$ as $u_{i}(\pi,\theta,r)$, then the expected utility across rounds can be defined as:
\begin{equation}
    U_{i}(\pi,\theta) = \sum_{r=0}^{\infty}\delta^{r}u_{i}(\pi,\theta,r)
\end{equation}

%%%%%%%%%%%%%%%%%%%%%%%%%%%%%%%%%%%%%%%%%
\subsection{Threat Model}
\label{ssec:network-model}
%%%%%%%%%%%%%%%%%%%%%%%%%%%%%%%%%%%%%%%%%
We model threat via $\mathcal{M}$ where $|T| = t\le t_0 $ and $|K| = k$. Here $t_{0}$ is the upper-bound on byzantine players to ensure security against byzantine-only attacks. 

We make a simple observation under the threat model $\mathcal{M}$, for any $t_{0} \geq 1$, the necessary condition for any protocol $\Pi$ to reach agreement with the threshold $\tau \in [\lfloor \frac{n + t_{0}}{2} \rfloor + 1, n - t_{0}]$. This threshold is such that $n - t_{0}$ players should agree on a value for agreement. 

\begin{claim}\label{claim:byz-lim}
    A protocol $\Pi$ achieves consensus with agreement of at least $\tau$ players agreeing on the same value under threat model $\mathcal{M} := \langle(\mathcal{P},T,K),\theta,t_{0}\rangle$ only if $\tau \in [\lfloor \frac{n + t_{0}}{2} \rfloor + 1, n - t_{0}]$
\end{claim}
\begin{proof}
    We prove the contrapositive of this claim. Consider the two cases where $\tau > n - t_{0}$ and $\tau \leq \lfloor \frac{n + t_{0}}{2} \rfloor$. If $\tau > n - t_{0}$ then a message/vote from at least one byzantine player is required to reach a consensus. Under this situation, each byzantine player can play $\pi_{abs}$ which would compromise the $(t,k)$\textsf{-eventual liveness} property of the protocol $\Pi$. If $\tau \leq \lfloor \frac{n + t_{0}}{2} \rfloor$ then consider the existence of network partition such that two subsets of players $A$ and $B$ are unable to communicate with each other except through set of adversaries $T$. Here, $A \cup B = \mathcal{P}\setminus T$, $A \cap B = \emptyset$ and $|A| = |B| = \frac{n - t_{0}}{2}$. If the leader in some round is $\mathcal{P}_{l} \in T$ then leader proposes $v_{a}$ to $A$ and $v_{b}$ to $B$. Since $|A| + |T| \geq \lfloor \frac{n - t_{0}}{2}\rfloor + t_{0} \geq \tau$ and similarly $|B| + |T| \geq \lfloor \frac{n - t_{0}}{2} \rfloor + t_{0} \geq \tau$, both partitions reach consensus on conflicting values which invalidates the $(t,k)$\textsf{-agreement} property. Therefore, $\Pi$ is not $(t,k)$\textsf{-robust} in either case.
\end{proof}

\subsection{Equilibrium \& Incentive Compatibility}
\label{ssec:eq-ic}
\paragraph{Nash Equilibrium (NE)} While analysing the rational security of protocols, a protocol is considered secure if following the protocol honestly is a \emph{Nash Equilibrium} strategy. Therefore, the aim is to design a protocol that is Nash Incentive Compatible. This means that following the protocol honestly is the \emph{nash-equilibrium} strategy for all rational players. We define Nash Incentive Compatible protocol as follows:
\begin{definition}[Nash Incentive Compatible]
    A protocol $\Pi$ with set of $K$ rational players is \emph{Nash Incentive Compatible} (NIC) for a given utility structure $U$ if $\forall i\in K$ in the set of rational players $K$ following strategy $s_{i}$ following the honest strategy $\pi_{0}$ is Nash Equilibrium. i.e. $\forall i \in K, \forall \pi$ 
    % \sg{$K$ is used first time here...i guess Sec 3.3 should be moved to Sec 4...and prev to subsections, explain altrusstic n rational players before 3.1 informally n say formal definitions are provided in Sec 4. }
    \[
    U_{i}(s_{i} = \pi,\{s_{j} = \pi_{0}\}_{j \neq i}) \leq U_{i}(s_{i} = \pi_{0},\{s_{j} = \pi_{0}\}_{j \neq i}) 
    \]
\end{definition}

\paragraph{Focal Point} Given a game, it could have multiple equillibria. Amongst, multiple equillibria,  a particular equilibrium may attract more attention than other equllibria. Such equilibrium is often referred as \emph{focal point}~\cite{Schelling1963SoC}. Consider the following example game between three players $P1$ having strategy space $\{A,B\}$, $P2$ having $\{a,b\}$ and $P3$ having $\{\alpha,\beta\}$. The utility is as given in Table~\ref{tab:utility-example} given in the order $(U_{P1},U_{P2},U_{P3})$. 

% insert table here...
\definecolor{Gray}{gray}{0.9}
\begin{table}[!th]
\centering
\begin{small}
\begin{tabular}{ p{2em} | p{4em}  p{4em} |  p{4em} p{4em}}
\toprule
\multirow{2}{*}{} & \multicolumn{2}{c}{a} & \multicolumn{2}{c}{b}\\
 & $\alpha$ & $\beta$ & $\alpha$ & $\beta$ \\
\midrule
\rowcolor{Gray}
$A$ & $(1,1,1)$ & $(1,1,0)$ & $(1,0,1)$ & $(-2,2,2)$ \\
$B$ & $(0,1,1)$ & $(1,-2,1)$ & $(2,2,-2)$ & $(0,0,0)$ \\
\bottomrule
\end{tabular}
\smallskip
\caption{Example of two equilibria in a $3$-player game (Here, utility is in the order $(A,a,\alpha)$)}
\label{tab:utility-example}
\end{small}
\end{table}

The game has two Nash equillibria, --- $(B,b,\beta)$ and  $(A,a,\alpha)$. The latter is attractive as it offers higher utility to all the players. Such focal points are important in analyzing a security game.

% \sg{my problem with the following para, we have not introduced trap  n baiting yet}
\paragraph{Challenges with multiple Nash Equilibria in a Security Game}
In case there are multiple NEs of a security game, if one of them implies security, does not imply security in general. The attackers would take the game towards an insecure equilibrium or rational players may play strategies resulting in equilibrium with higher utilities to them. Thus, we must explore all equlibiria and ensure security at the worst equilibrium. %  may exist for a security game. We should not conclude security based on a single secure NE. Rather, we should investigate the existence of insecure equilibria that may be more attractive to many players. 

In our analysis, we are going to argue that in the generalized model in this paper,  there are multiple equilibria and baiting-based equilibrium~\cite{Gramoli2022Trap} that ensure the security of the protocol is one of them (similar to $(B,\beta,b)$ in Table~\ref{tab:utility-example}). We show in Theorem~\ref{thm:trap-impossible}  that there is another equilibrium (similar to $(A,\alpha,a)$) which may be more attractive to rational players. At the later equilibrium, the protocol is not secure. Thus, in the generalized model, TRAP~\cite{Gramoli2022Trap} need not be secure. 

\if 0
In this table, although points $(B,\beta,b)$ is an equilibrium point, a more stable equilibrium point is $(A,a,\alpha)$. In Ranchal-Pedrosa \& Gramoli's TRAP~\cite{Gramoli2022Trap} $(B,\beta,b)$ corresponds to baiting-based equilibrium point which ensures the security of the protocol, but by Theorem~\ref{thm:trap-impossible} we show the existence of another more insecure equilibrium point (analogous to $(A,a,\alpha)$ in our example). Therefore arguing against the solution proposed in~\cite{Gramoli2022Trap}. More details about the specifics of the protocol can be found in Section~\ref{sec:impossibilities}. In security games, our goal should be that following the protocol should such a stable Nash equilibrium. We consider an equilibrium point $X$ more stable than $Y$ if decrease in utility is lesser for an agent when other agents deviate from equilibrium point $X$ than $Y$. While the stability of eqilibrium points is not primary focus of this work, it motivates that in works like~\cite{Gramoli2022Trap}, while the proposed protocol is Nash incentive compatible, there exists a better equilibrium point that incentivizes rational players to deviate from the protocol. We formally prove the existence of this equilibrium point in Theorem~\ref{thm:trap-impossible}. 
\fi

\paragraph{Dominant Strategy Equilibrium (DSE).} A better equilibrium from the weaker \emph{stable Nash equilibrium} is a DSE. DSE equilibrium points are not contested by other equilibrium points. Thus, we can safely assume rational agents follow a DSE strategy, and the protocol is Dominant Strategy Incentive Compatible (defined below).

\begin{definition}[Dominant Strategy Incentive Compatible]
    A protocol $\Pi$ with a set of $K$ rational players is \emph{Dominant Strategy Incentive Compatible} (DSIC) for a given utility structure $U$ if $\forall i\in K$ following honest strategy $s_{i} = \pi_{0}$ is Nash Equilibrium. i.e. $\forall i \in K, \forall \pi,\forall s_{j}\forall j \in K/\{i\}$ 
    \[
    U_{i}(s_{i} = \pi,\{s_{j}\}_{j \neq i}) \leq U_{i}(s_{i} = \pi_{0},\{s_{j}\}_{j \neq i}) 
    \]
\end{definition}

\noindent We propose \proname\ in Section~\ref{sec:our-protocol} which is DSIC; therefore providing a better security guarantee.

%%%%%%%%%%%%%%%%%%%%%%%%%%%%%%%%%%%%%%%%%
\subsection{Impossibilities}
\label{sec:impossibilities}
%%%%%%%%%%%%%%%%%%%%%%%%%%%%%%%%%%%%%%%%%

From the discussion of the previous section, rational players can be of different types $\theta \in \{0,1,2,3\}$. We also know through Claim~\ref{claim:byz-lim} that any protocol requires $\tau \in [\lfloor \frac{n + t_{0}}{2}\rfloor + 1,n - t_{0}]$ for security against \emph{byzantine} attacks (Note that, this is necessary but not sufficient). We show through Theorem~\ref{thm:r3-impossible} that under a stricter (more adversarial) agent type for rational players, achieving RC for any $k + t > \frac{n}{3}$ is impossible.

\begin{theorem}[Rational Consensus under $\theta = 3$]\label{thm:r3-impossible}
    Under the threat model $\langle(\mathcal{P},T,K),\theta=3,t_{0}\rangle$ no rational consensus protocol is $(t,k)$\textsf{-robust} when $\lceil\frac{n}{3}\rceil \leq k + t \leq \lceil\frac{n}{2}\rceil - 1$.
\end{theorem}
\begin{proof}
    The proof shows that for any arbitrary protocol $\Pi$, rational players are incentivized to follow $\pi_{abs}$. Since $\pi_{abs}$ is indistinguishable from crash faults, no penalty-based (accountable) protocol can distinguish this behaviour (and thus reduce the utility). Due to space constraints, we omit the inclusion of the full proof in the paper. This compromises $(t,k)$\textsf{-eventual liveness} property and therefore, any protocol $\Pi$ is not a $(t,k)$\textsf{-robust} RC protocol.
    \iffalse
    Consider a protocol $\Pi$ belonging to the set of protocols $\mathcal{C}_{3}$ that achieve consensus under threat model $\mathcal{M} = \langle(\mathcal{P},T,K),\theta=3,t_{0}\rangle$. Let $n = |\mathcal{P}|$ therefore, $|T| + |K| < \frac{n}{2}$. In such a case, consider the threshold of messages required for agreement by the protocol to be $\tau$. From Claim~\ref{claim:byz-lim} we have $\tau > \lfloor \frac{n + t_{0}}{2} \rfloor \geq \lceil \frac{n}{2} \rceil$ (since $t_{0} \geq 1$). This means that for consensus, the message/signature of at least one player in $K \cup T$ is required. In this case, if each player $\mathcal{P}_{l} \in K \cup T$ follows strategy $\pi_{abs}$ -- not sending any messages, then consensus is not reached. Since abstaining from sending messages in a round is indistinguishable from message delays due to partially-synchronous networks, $\pi_{abs}$ cannot be distinguished from $\pi_{0}$ under partially-synchronous network. Therefore, $D(\pi_{abs},\sigma) = 0$. The utility for following $\pi_{abs}$ for each rational player is
    \[
        \begin{aligned}
        U_{i}(\pi_{abs},\theta=3) = & \sum_{r=0}^{\infty} \delta^{r}u_{i}(\pi_{abs},3,r) \\ 
        = & E_{\sigma\sim\;S}[f(\sigma,3)] - 0 \\
        = & \alpha > 0 = U_{i}(\pi_{0},\theta=3)
        \end{aligned}
    \]
    Therefore, each player in set $K$ is incentivized to play $\pi_{abs}$ over following protocol which compromises $(t,k)$\textsf{-eventual liveness} property and therefore any such arbitrary $\Pi$ is not a $(t,k)$\textsf{-robust} RC protocol.
    \fi
\end{proof}

We also show a pessimistic result for rational players of $\theta = 2$. We show through Theorem~\ref{thm:r2-impossible} that for every protocol if the rational players are of type $\theta = 2$ then there always exists a strategy which compromises  $(t,k)$\textsf{-censorship resistance} while ensuring  $(t,k)$\textsf{-eventual liveness}. Note that this impossibility holds despite of existence of threshold encryption schemes~\cite{Miller2016Honeybadger},

\begin{theorem}[Rational Consensus under $\theta=2$]\label{thm:r2-impossible}
    Under threat model $\langle(\mathcal{P},T,K),\theta = 2,t_{0}\rangle$ no rational consensus protocol is \textsf{strongly} $(t,k)$\textsf{-robust} when $\lceil\frac{n}{3}\rceil \leq t + k \leq \lceil\frac{n}{2}\rceil - 1$.
\end{theorem}
\begin{proof}

The proof follows by showing the existence of a strategy $\pi_{pc}$ where adversarial collusion set follows (1) $\pi_{abs}$ when the leader is honest, (2) $\pi_{0}$ (and omit censored transactions) when the leader is adversarial. Due to the indistinguishability between crash faults and $\pi_{abs}$, no accountable protocol is possible. Due to space constraints, the full proof is provided in Appendix~\ref{app:r2-impossible}.
\end{proof}

Having proven the impossibility of having a protocol that achieves RC for rational player types $\theta = 3$ and $2$, we now relax the player type further to $\theta = 1$. Under this utility model for rational players, Gramoli et al.~\cite{Gramoli2022Trap} proposed a Baiting-based protocol (which they call TRAP) and show Baiting is \emph{necessary} and \emph{sufficient} to achieve RC under partially-synchronous network settings for $t + k < \frac{n}{2}$. Following our brief discussion of the model and result of~\cite{Gramoli2022Trap} in Section~\ref{ssec:trap}, we show that the Baiting-based protocol does not ensure RC in repeated rounds of consensus due to the existence of another Nash equilibrium point which leads to disagreement. Further, this point being focal equilibria (see discussion in Section~\ref{ssec:eq-ic}), it will be preferred over the secure equilibria point proposed by~\cite{Gramoli2022Trap}. We demonstrate this result in Theorem~\ref{thm:trap-impossible} below. For notational consistency of this result with~\cite{Gramoli2022Trap} we use $R := f(\sigma_{fork},1)$

\begin{theorem}[Baiting based Rational Consensus under $\theta = 1$]\label{thm:trap-impossible}    
    Consider any baiting-based RC protocol $\Pi$ the threat model $\mathcal{M} = \langle (\mathcal{P},K,T),\theta=1,t_{0}\rangle$. The set of rational players following $\pi_{fork}$ is a Nash-equilibrium for $|K| > 2 + t_{0} - t$. Thus, $\Pi$ is not $(t,k)$\textsf{-robust} RC for $t_{0} = \lceil\frac{n}{3}\rceil - 1$.  
\end{theorem}
\begin{proof}
    The proof follows from showing that if collusion follows grim-trigger strategy\footnote{grim-trigger: if one player of collusion baits, all players will abandon collusion.}, then another Nash Equilibria (NE) exists in repeated rounds where players of the collusion deviate in every round. This Equilibrium is pareto-optimal due to which players are more prone to end up at this NE point than the secure NE proposed by~\cite{Gramoli2022Trap}.
\end{proof}  
% \begin{table}[!t]
% \begin{small}
% \centering
% \begin{tabular}{ p{5em} | p{3em} p{3em}  p{3em}}
% \toprule
% \multirow{2}{*}{$t_{0}$} & \multicolumn{3}{c}{$|K| + |T|$}\\
%  & $[0,\frac{n}{3})$ & $[\frac{n}{3},\frac{n}{2})$ & $[\frac{n}{2},n]$ \\
% \midrule
% \multirow{2}{*}{$\left(t_{0} < \frac{n}{2}\right)^{\dagger}$} \\
% & $\checkmark$ & $\checkmark$ & $\times$ \\
% \hline
% \multirow{2}{*}{$t_{0} < \frac{n}{2}$} \\
% & $-$ & $\times$ & $\times$ \\
% \hline
% \multirow{2}{*}{$t_{0} < \frac{n}{3}$}\\
% & $\checkmark$ & $\times$ & $\times$ \\
% \hline
% \multirow{2}{*}{$t_{0} < \frac{n}{4}$} \\
% & $\checkmark$ & $\checkmark$ & $\times$ \\
% \bottomrule
% \end{tabular}
% \smallskip
% \\
% ($\dagger$: synchronous network)
% \caption{Possibility of RC for $(\mathcal{P},K,T)$}
% \label{tab:messages}
% \end{small}
% \end{table}

In this section, we demonstrated the impossibility of achieving RC when rational players are of the type $\theta = 3$ or $2$ and $ k + t > \frac{n}{3}$. We further show that when rational players are of type $\theta=1$, the previously existing ``baiting-based'' solution~\cite{Gramoli2022Trap} has a Nash equilibrium strategy, resulting in Disagreement. From our discussion in Section~\ref{ssec:eq-ic}, this is a more stable equilibrium point than the secure equilibrium point proposed by~\cite{Gramoli2022Trap} (that relies on $m > \frac{t + k - n}{2} + t_{0}$ deviating together; ref. Lemma 4.3 and Lemma 4.4 in~\cite{Gramoli2022Trap}). In the next section, we propose a protocol \proname\ which achieves $(t,k)$\textsf{-robust} RC without relying on baiting by rational players. 

%%%%%%%%%%%%%%%%%%%%%%%%%%%%%%%%%%%%%%%%%
\section{\textsf{pRFT}: Rational Consensus Protocol}
\label{sec:our-protocol}
%%%%%%%%%%%%%%%%%%%%%%%%%%%%%%%%%%%%%%%%%

Baiting-based consensus~\cite{Gramoli2022Trap} introduces interesting insights about using Proof-of-Fraud (PoF) to penalize deviating rational players. The protocol relies on incentivizing rational players to bait the collusion, and we show in Theorem~\ref{thm:trap-impossible} existence of another (better) Nash Equilibrium for any rational player to deviate to $\pi_{Fork}$ than follow honest-baiting $\pi_{bait}$ strategy. To solve this problem, we attempt to capture PoF through honest players prone to any incentive manipulation. Towards this, we propose \proname, which is described below.

\subsection{Protocol}
The \proname\ protocol runs in discrete rounds. The set of players involved are $\mathcal{P} = \{\mathcal{P}_{1},\mathcal{P}_{2}\ldots,\mathcal{P}_{n}\}$ and in round $r$ the leader is $\mathcal{P}_{l}$ for $l = 1 + (r\;\text{mod}\;n)$. We assume the network is partial-synchronous, and reliable-broadcast i.e. messages reach to the receiver untampered, although they might suffer network delays. Note that for brevity, we abstract the cryptographic verification of the message to be done by the Recv procedure (lines $7,11,17,24,27$ and $29$ of the protocol in Figure~\ref{fig:protocol}). Therefore, any message coming through it will contain only valid signatures, and invalid messages are discarded.
Each round progresses in $4$ distinct phases, described as follows:

\paragraph{Propose Phase} The leader $\mathcal{P}_{l}$ has a set of transactions that they want to publish to the ledger (blockchain). $\mathcal{P}_{l}$ selects a set of these transactions $\overline{tx} = \{tx_{1},tx_{2}\ldots,tx_{c}\}$ and form a \textsf{Block} $B_{r}$. She then proposes this block by broadcasting over the network via a \emph{propose} message with their cryptographic signature $s_{l}$ on the hash $h_{l}:= \mathbf{H}(\textsf{Block}||r)$ of the \textsf{Block}\footnote{$h_{l}$ also contains the round number therefore, signed messages from one round can not be used in another round using replay attack.}. They construct message $m := (\langle Propose, B_{l}, h_{l}, r\rangle$ and signature $s^{pro}_{l}$ to broadcast $(m,s_{l}^{pro})$. All non-leader players $\mathcal{P}_{i}$ receive the propose message from the leader and move to the vote phase. The non-leader players $\mathcal{P}_{i} \;\forall i\not=l$ (1) check the validity of the \emph{propose} message from the leader $\mathcal{P}_{l}$ (2) broadcast a \emph{vote} message if the \emph{propose} message is valid. In checking the validity of the \emph{propose} message, the player verifies $H_{l} = \mathbf{H}(\textsf{Block}||r)$ and verifies signature $s_{l}$ on the message $\langle Propose, B_{l}, h_{l}, r\rangle$. They then sign $s_{i}^{vote}$ on the message $m^{vote}_{i} := \langle Vote,h_{i},,s^{pro}_{l},r\rangle$ and broadcast $(m^{vote}_{i},s_{i}^{vote})$ over the network.
% They construct the \emph{vote} message, which contains the hash of the block  proposed by the leader, the leader's signature, and the message (hash of the block) is signed by the voting player $\mathcal{P}_{i} \in \mathcal{P}$. A vote message is considered \emph{valid} if (i) the sender's signature on the \textsf{Block} hash is valid and (ii) the leader's signature on the same \textsf{Block} hash (which is also contained in the message) is valid (iii) round number $r$ is the current round. 

\paragraph{Vote Phase} In this phase, they wait for $n - t_{0}$ \emph{valid} vote messages from other players from the previous phase for some (same) proposed value $h_{*}$. If no such value is obtained, the player sets $h_{*} := \perp$ (default empty value). Each player commits to the value $h_{*}$ by constructing a Commit message. This consists of the decided value $h_{*}$ and set $V_{i}$ of $\geq n - t_{0}$ votes on this value $V_{i} = \emptyset$ if $h_{*} = \perp$. The player $\mathcal{P}_{i}$ signs on the message $m_{i}^{com} := \langle Commit, h_{*}, s^{pro}_{l}, V_{i}, r\rangle$ and broadcasts $(m_{i}^{com},s_{i}^{com})$ to all other players. 

\newcommand{\TAB}[0]{\hspace{15pt}}
\begin{figure}[t]
    \centering
    \begin{small} 
    \begin{mybox}{\proname($\overline{p}_{i=1}^{n},t_{0}$)}
    \begin{multicols}{2}
    \begin{algorithmic}[1]
    \STATE\textbf{Propose Phase:}\\
    \IF{$i = r \text{ mod } n$} 
        \STATE $\textsf{Block}_{r}$ := \textsf{ConstructBlock}$(\overline{tx},r,b_{parent},p_{i})$\STATE $h_{l} := \textsf{Hash}(\textsf{Block}_{r})$
        \STATE \textsf{Broadcast}$(\langle Propose, Block_{r}, h_{l}, r\rangle_{i},s^{pro}_{i})$
    \ELSE
        \STATE On Recv. $(\langle Propose, Block_{r},h_{l},r \rangle, s_{l}^{pro})$:
        \STATE \TAB \textsf{Broadcast}$(\langle Vote, h_{l}, s_{l}^{pro},r\rangle_{i},s_{i}^{vote})$
    \ENDIF
    \vspace{5pt}
    \STATE \textbf{Vote Phase:}
    \STATE On Recv. $(\langle Vote, h_{j}, s_{l}^{pro},r\rangle,s_{j}^{vote})$:
    \STATE \TAB $votes[h_{j}] := votes[h_{j}] \cup \{s_{j}^{vote}\}$
    \IF{for some $h_{*}, vote[h_{*}].size \geq n - t_{0}$}
    \STATE \textsf{Broadcast}\\$(\langle Commit, h_{*}, s_{l}^{pro}, vote[h_{*}],r\rangle_{i},s_{i}^{com})$
    \ENDIF
    \vspace{5pt}
    \STATE \textbf{Commit Phase:}
    \STATE On Recv. $(\langle Commit, h_{j}, s_{l}^{pro},\overline{vote}_{j},r\rangle,s_{j}^{com})$:
    \STATE \TAB $commit[h_{j}] := commit[h_{j}] \cup \{s_{j}^{com}\}$
    \IF{for some $h_{*}, commit[h_{*}].size \geq n - t_{0}$}
    \STATE \textsf{Broadcast}$(\langle Reveal, h_{*}, s_{l}^{pro}, commit[h_{*}],r\rangle_{i},s_{i}^{rev})$
    \ENDIF  %\hspace{1cm}\COMMENT{End of Commit Phase}
    \pagebreak
    \vspace{5pt}
    \STATE \textbf{Reveal Phase:}
    \STATE $D_{i} := \emptyset, M_{i} := \emptyset, F_{i} := \emptyset$
    \STATE On Recv. $(\langle Reveal, h_{j}, s_{l}^{pro},\overline{commit}_{j},r\rangle,s_{j}^{rev})$:
    \STATE \TAB $M_{i} \gets M_{i} \cup \{\overline{commit}_{j}\}$
    \STATE \TAB $D_{i} := \textsf{ContructPoF}(M_{i})$
    \STATE On Recv $(\langle Final,B_{j},s^{pro}_{l}\rangle_{j},s_{j})$
    \STATE \TAB $F_{i} \gets F_{i} \cup \{s_{j}^{fin}\}$
    \STATE On Recv $(\langle Expose, D_{j}, r\rangle,s_{j})$
    \STATE \TAB \textsf{Stash}$(D_{j})$, $r := r + 1$ 
    \IF{$|D_{i}| > t_{0}$}
    \STATE \textsf{Broadcast}($\langle Expose, D_{i},r\rangle_{i},s_{i}^{expose}$)
    \ELSIF{$|M_{i}| \geq n - t_{0}$}
    \STATE \textsf{Broadcast}$(\langle Final,B_{l},s^{pro}_{l}\rangle_{i},s_{i}^{fin})$
    \ELSIF{$|F_{i}| > \frac{n}{2}$}
    \STATE \textsf{Broadcast}$(\langle Final,B_{l},s^{pro}_{l}\rangle_{i},s_{i}^{fin})$
    \ENDIF
    
    \end{algorithmic}
    \end{multicols}
    \end{mybox}
    \end{small}
    \caption{\proname\ Protocol}
    \label{fig:protocol}
\end{figure}

\paragraph{Commit Phase} Upon receiving $\geq n - t_{0}$ commit messages for a particular (same) value $h_{tc}$, the player reaches \emph{tentative-consensus} on this block. Each player $\mathcal{P}_{i}$ shares their tentative consensus by sharing value $h_{tc}$ which got $\geq n - t_{0}$ commit messages and a Proof-of-Commitment which is the set $W_{i}$ (for player $\mathcal{P}_{i}$ of $n - t_{0}$ signatures on the commit messages on value $h_{*}$. They construct message $m_{i}^{rev} := \langle Reveal, h_{tc},h_{l}, W_{i}, r\rangle$ and signature $s_{i}^{rev}$ on this message to broadcast $(m_{i}^{rev},s_{i}^{rev})$.

\paragraph{Reveal Phase} Each player verifies across the set of Proof-of-Commitments $W_{j}$ for each $(m_{j}^{rev},s_{j}^{rev})$ received by $\mathcal{P}_{i}$ any attempts of \emph{double-signature} in the proof vectors. Each player accumulates a set of double signatures as Proof-of-Fraud (PoF) by invoking the \textsf{ConstructProof} procedure (see Figure~\ref{fig:pof-protocol}). These PoF can be used to \emph{burn} the tokens/coins deposited by the deviating player by including a corresponding transaction in a future block. If there are $\geq n - t_{0}$ messages and a total of $\leq t_{0}$ players with double signatures (PoF) then the player reaches \emph{Final Consensus} on the proposed block $B_{l}$. In this case, they construct a message $m^{fin}_{i} := \langle Final, h_{l},s_{l}^{pro} \rangle$ and signature $s_{i}^{fin}$ on it. They then broadcast this pair $(m_{i}^{fin},s_{i}^{fin})$ to the client/network. Otherwise, if the set of double signatures $\geq t_{0} + 1$, they broadcast this PoF with $\geq t_{0} + 1$ double-signatures (represented by set $D_{i}$). They then construct message $m_{i}^{exp} := \langle Expose, D_{i}, r\rangle$ and signature $s_{i}^{exp}$ and broadcast $(m_{i}^{exp},s_{i}^{exp})$. If the player obtains $> \frac{n}{2}$ Final messages, this means at least one honest player has finalized on the block. Then, this player also finalizes and broadcasts a final message.

\begin{figure}
\centering
\begin{subfigure}[c]{.45\textwidth}
  \includegraphics[width=0.9\linewidth]{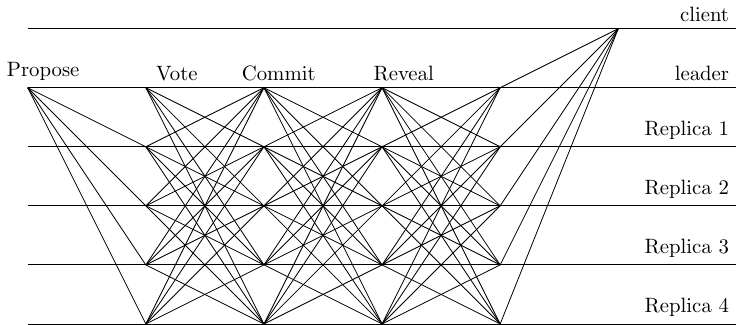}
  \caption{Normal execution of \proname\ protocol}
  \label{fig:protocol-diagram}
\end{subfigure}%
\begin{subfigure}[c]{.4\textwidth}
    \begin{small}
    \begin{tabular}{ p{6em}  | p{12em}}
        \toprule
        Message Type  & Message \\
        \midrule
        \emph{propose} & $(\langle Propose, B_{l}, h_{l}, r\rangle,s_{l}^{pro})$\\
        \emph{vote} & $(\langle Vote,h_{i},,s^{pro}_{l},r\rangle,s_{i}^{vote})$\\
        \emph{commit} & $(\langle Commit, h_{*}, s^{pro}_{l}, V_{i}, r\rangle,s_{i}^{com})$\\
        \emph{reveal} & $(\langle Reveal, h_{tc},h_{l}, W_{i}, r\rangle,s_{i}^{rev})$\\ 
        \emph{expose} & $(\langle Expose, D_{i}, r\rangle,s_{i}^{exp})$\\
        \emph{final} & $(\langle Final, h_{l},s_{l}^{pro} \rangle,s_{i}^{final})$\\
        \midrule
        \emph{view-change} & $(\langle ViewChange,\textsf{Phase},r, \rangle,s_{i}^{vc})$\\
        \emph{commit-view} & $(\langle CommitView,V_{i},r\rangle,s_{i}^{cv})$\\
        \bottomrule
    \end{tabular}
    \smallskip
    \end{small}
    \caption{Messages sent in \proname}
    \label{tab:messages}
\end{subfigure}
\caption{\proname\ protocol execution \& messages}
\label{fig:test}
\end{figure}

% \smallskip \noindent In case of a timeout/malicious behaviour, a two-phase \emph{view-change} is triggered to change the leader. 

%%%%%%%%%%%%%%%%%%%%%%%%%%%%%%%%%%%%%%
\subsection{View Change} 
\label{ssec:view-change}
%%%%%%%%%%%%%%%%%%%%%%%%%%%%%%%%%%%%%%%

In each round players wait for proposal messages in the proposed phase or $\geq n - t_{0}$ messages in the other phases. Either due to delays in the network or $> t_{0}$ players deviating from the protocol, if a timeout happens (when the duration of phase exceeds the local waiting time of $\Delta$) then \emph{view change} is initiated. The view change protocol proceeds as follows: 
\begin{itemize}[leftmargin=8pt]
    \item[1] A player triggers view-change if the following happens:
        \begin{itemize}[leftmargin=8pt]
            \item[-] a timeout in waiting time $\Delta$
            \item[-] conflicting signatures on 2 different proposed values $v_{l1},v_{l2}$ in the same round by the leader $\mathcal{P}_{l}$
            \item[-] conflicting signatures by $\geq t_{0} + 1$ players in some phase.
        \end{itemize}
    Under either of these scenarios, the player signs the message $m_{i}^{vc} := \langle View-Change,\textsf{Phase},r, \rangle$ with signature $s_{i}^{vc}$ and broadcast $(m_{i}^{vc},s_{i}^{vc})$ to the network.
    \item[2] If a player $\mathcal{P}_{i}$ receives a view-change message from player $\mathcal{P}_{j}$ for some phase in some round, they store the message (1) Wait for $\geq n - t_{0}$ such messages (from the same phase) or timeout or, (2) If they have $\geq n - t_{0}$ messages from that phase, they send the corresponding messages to $\mathcal{P}_{j}$.
    \item[3] If a player receives $\geq n - t_{0}$ \emph{view-change} messages including their own \emph{view-change} (represent this set as $V_{i}$), they construct a \emph{commit-view} message $m_{i}^{cv} := \langle CommitView,V_{i},r\rangle$ along with the signature $s_{i}^{cv}$. The player will discontinue the current round and wait for a view change.
    \item[4] If a player receives a commit-view message, they verify if it consists of $n - t_{0}$ valid (signed) view-change messages. If so, they commit to view change and send a \emph{commit-view} message. 
    \item[5] When the player receives $> n - t_{0}$ commit-view messages, they commit to view change and change round from $r$ to $r + 1$ and begin the corresponding propose phase.
\end{itemize}

The view change sub-protocol should ensure the following two properties:
\begin{itemize}[leftmargin=8pt]
    \item \textbf{Consistency:} If a player $\mathcal{P}_{i} \in H$ has committed to view-change, then any other player $\mathcal{P}_{j} \in H$ should also eventually commit to view-change (should not reach agreement in $r$).
    \item \textbf{Robustness:} The set $T$ cannot launch a view-change if the leader is honest $\mathcal{P}_{l} \in H$.
\end{itemize}
We show in the following Claim~\ref{claim:view-change-properties} that the view-change protocol of \proname\ satisfies both of these properties. We defer the proof to Appendix~\ref{app:view-change-properties}.

\begin{claim}\label{claim:view-change-properties}
    The \emph{view-change} sub-protocol of \proname\ satisfies both \emph{Consistency} and \emph{Robustness}.
\end{claim}

\subsection{Discussion}

The protocol leverages Proof-of-Fraud (PoF) to penalize rational players on the following $\pi_{ds}$. In addition, similar to~\cite{Micali2017Algorand} \proname\ uses \emph{tentative} and \emph{final} consensus. We discuss (i) how PoF is realized in \proname, (ii) the advantage of tentative and final consensus, and (iii) message complexity of \proname. 

\subsubsection{Accountability and Proof-of-Fraud (PoF)}
\label{sssec:proof-of-fraud}
\proname\ implements penalty mechanism via Proof-of-Fraud. Each player that is a part of the consensus committee $\mathcal{P}$ deposits some amount $L$ as collateral.  This collateral is locked unless some specified $q$ number of blocks are mined. If there is some malicious behaviour by player $\mathcal{P}_{i}$, and a PoF exists against them, this PoF can be used as an input to the transaction to burn the collateral $L$ of the player $\mathcal{P}_{i}$. Due to space constraints, we formally present the PoF construction in Appendix~\ref{app:construct-proof}. The property where deviation of more than $t_{0}$ (for some $t_{0}$) players is captured along with the identities of deviating players is called accountability and some existing consensus protocols provide accountability~\cite{Civit2021Polygraph,Gramoli2022Trap,Ranchal2020ZLBAccountable} of deviating players. Motivating from~\cite[Definition 1]{Civit2021Polygraph} we define Accountability as follows:
\begin{definition}[Accountability]\label{def:accountability}
    If two honest parties output different values, then eventually all honest parties reach a state $s_j$ and receive/construct some Proof-of-Fraud (PoF) $\pi \in \{0,1\}^{*}$ such that $\exists$ verification algorithm $V(\cdot)$ the value $V(\pi)$ outputs set of $\geq t_{0} + 1$ deviating (guilty) players.
\end{definition}

\definecolor{DarkGreen}{rgb}{0,0.4,0}
\begin{figure}
\centering
    \begin{small}
    \begin{tabular}{ p{14em} | p{7em}  p{5em} p{6em}}
    \toprule
    \textbf{Protocol} & \textbf{Message Complexity} & \textbf{Message Size} & \textbf{Accountability} \\
    \midrule
    \rowcolor{Gray}
    pBFT~\cite{CastroLiskov1999pBFT} & $O(n^{3})$ & $O(\kappa\cdot n^{4})$ &$\times$ \\
    Hotstuff~\cite{Yin2019Hotstuff} & $O(n^{2})$ & $O(\kappa\cdot n^{3})$ & $\times$ \\
    \rowcolor{Gray}
    Polygraph~\cite{Civit2021Polygraph}$^{\dagger}$ & $O(n^{3})$ & $O(\kappa\cdot n^{4})$ & \checkmark \\ 
    \color{DarkGreen}\proname & \color{DarkGreen}$O(n^{3})$ & \color{DarkGreen}$O(\kappa\cdot n^{4})$ & \color{DarkGreen}\checkmark \\ 
    \bottomrule
    \end{tabular}
    
    $\dagger$While polygraph achieves same guarantees, their threat model is much weaker than \proname's
    \smallskip
    \caption{Message Complexity for different consensus protocols compared with \proname}
    \label{tab:message-complexity}
    \end{small}
\end{figure}

\subsubsection{Tentative and Final Consensus}
\label{sssec:tentative-final-consensus}
After the correct execution of phases 1-3 (propose, vote and commit) of the \proname, each player reaches a \emph{tentative consensus} on the \textsf{Block} $B_{r}$. $B_r$ reaches \emph{final consensus} after the correct execution of step 4 if no $\mathcal{P}_{i} \in K$ deviates from the protocol. 

\paragraph{Effectiveness of Tentative Consensus} It is interesting to observe that tentative consensus will be finalized unless a rational player deviates from the protocol. For rational players of type $\theta = 1$ they are incentivized only in system state $\sigma_{Fork}$ for which they would have to follow $\pi_{ds}$. However, as we show in the following Section~\ref{sec:protocol-analysis}, this deviation can be caught in phase 4 of the protocol. Due to this, any $\mathcal{P}_{i} \in K$ is disincentivized from deviating from the protocol, ensuring that Tentative Consensus will convert to Final Consensus. 

\subsubsection{Message Complexity}
In normal executions, \proname\ uses all-to-all broadcasts in $4$ phases, leading to the message complexity $n^{3}$. Additionally, in Vote, Commit and Reveal phases, they share a set of signatures on messages from previous rounds. Therefore, the message complexity becomes $\kappa\cdot n^{4}$ for security parameter $\kappa$ used in PKI. The message complexity and size are on par with the best protocols that provide accountability, such as~\cite{Civit2021Polygraph} (while \proname\ tolerates a stricter adversary model compared to these solutions). We present this comparison in Table~\ref{tab:message-complexity} (from~\cite[Table 1]{Civit2021Polygraph}).

%%%%%%%%%%%%%%%%%%%%%%%%%%%%%%%%%%%%%%%%%
\section{Analysis of \textsf{pRFT}}
\label{sec:protocol-analysis}
%%%%%%%%%%%%%%%%%%%%%%%%%%%%%%%%%%%%%%%%%
We analyse the security and liveness of the protocol under partial synchronous network and threat model $\mathcal{M} = \langle (\mathcal{P},T,K),\theta=1,\lceil\frac{n}{4}\rceil - 1\rangle$. Therefore in the worst case, $|T| = t_{0}$ and $n = 4t_{0} + 1$. Further, $k + t < \frac{n}{2}$. First we show that irrespective of strategy followed by $\mathcal{P}_{i} \in T$ if remaining $\mathcal{P}\setminus T$ do not deviate from the protocol,  agreement is reached on one value (in periods of synchrony) when leader is non-deviating. In a period of asynchrony, view-change (due to timeout) happens. 

\begin{claim}\label{claim:agreement-byz}
    In any round $r$ of \proname\ such that leader $\mathcal{P}_{l} \not\in T$, irrespective of the strategy of the set $T$, if remaining $\mathcal{P}\setminus T$ play $\pi \not= \pi_{fork}$ then exactly one of the following holds: 
    \begin{itemize}[leftmargin=8pt]
        \item If network is synchronous\footnote{note that this period of synchrony can happen either if network is synchronous or GST event has already happened in partially-synchronous network} and $\mathcal{P}\setminus T$ play honestly, agreement is reached on a single block $B_{l}$
        \item If network is asynchronous\footnote{happens in partially-synchronous networks before GST event} or some $\mathcal{P}_{i} \in K$ follows $\pi \not\in \{\pi_{fork},\pi_{0}\}$, timeout triggers view-change. 
    \end{itemize}
\end{claim}
\begin{proof}
    The proof follows by showing that any arbitrary network partition of $\mathcal{P}/\mathcal{T}$ does not lead to $\geq 2$ disjoint subsets $A,B$ such that $A\cup T$ and $B \cup T$ are $\geq n - t_{0}$. Therefore, either timeout happens due to insufficient ($< n - t_{0}$) messages in all partitions or agreement is reached in exactly one partition. The complete proof is omitted due to space constraints.
\end{proof}

To prove that \proname\ realizes $(t,k)$\textsf{-robust} RC, we first show through Lemma~\ref{lemma:rational-honest} that any $\mathcal{P}_{i} \in K$ is disincentivized from deviating from the protocol i.e. following any $\pi \not= \pi_{0}$. We finally conclude in Theorem~\ref{thm:main-theorem} that \proname\ is a $(t,k)$\textsf{-robust} RC protocol under the threat model $\mathcal{M}$.

\begin{lemma}\label{lemma:rational-honest}
    For any $\mathcal{P}_{i} \in K$ under threat model $\mathcal{M} = \langle(\mathcal{P},K,T),\theta=1,\lceil\frac{n}{4}\rceil - 1\rangle$ and protocol \proname, following the protocol honestly (i.e. strategy $\pi_{0}$) is dominant strategy incentive compatible (DSIC) for $|K| + |T| < \frac{n}{2}$ and $t < t_{0}$. That is, $U_{i}(\pi_{0},1) \geq U_{i}(\pi,1) \hspace{10pt}\forall\;\pi,\;\forall\;\mathcal{P}_{i} \in K$.
\end{lemma}

Due to space constraints, we defer the proof to Appendix~\ref{app:rational-honest}. From Lemma~\ref{lemma:rational-honest}, we conclude that any rational player of type $\theta=1$ will follow \proname\ honestly under the threat model $\mathcal{M}$ as described above. Our discussion from Section~\ref{sssec:tentative-final-consensus} implies that if all players in $K$ behave rationally, then all tentative consensus will be converted to final consensus except during timeout due to network asynchrony. We now conclude in Theorem~\ref{thm:main-theorem} that \proname\ realizes $(t,k)$\textsf{-robust} RC. 
\begin{theorem}
\label{thm:main-theorem}    
    \proname\ is a $(t,k)$\textsf{-robust} rational consensus protocol under threat model $\mathcal{M} = \langle(\mathcal{P},T,K),\theta=1,\lceil\frac{n}{4}\rceil - 1\rangle$ for $|K| + |T| < \frac{n}{2}$ under synchronous and partial-synchronous networks. 
\end{theorem}
\begin{proof}
    From Lemma~\ref{lemma:rational-honest} we know all $K$ follow $\pi_{0}$. All $H$ follow $\pi_{0}$ by definition of Honest players. Therefore, only deviating players are from the set $T$. From Claim~\ref{claim:agreement-byz} we know that in such a scenario, one of two things happen -- view-change (if the network is in a period of asynchrony), or agreement (if all messages from a particular phase has been delivered. In synchronous settings, agreement is trivially satisfied due to Claim~\ref{claim:agreement-byz}. In partially-synchronous setting, all messages from a round are eventually delivered before the next GST. Therefore, some block $B_{r}$ (and therefore all blocks before $B_{r}$) are finalized  during the period of synchrony. Since $n - t_{0}$ are following $\pi_{0}$ (and not trying to cause censorship attack), \proname\ is a \textsf{strongly} $(t,k)$\textsf{-robust} RC protocol.
\end{proof}

To summarize, we have shown through Lemma~\ref{lemma:rational-honest} that following honest strategy is DSIC for all players $\mathcal{P}_{i} \in K$ in partially-synchronous network. Together with Claim~\ref{claim:agreement-byz} this means under synchronous and partially synchronous networks, \proname\ is \textsf{strongly} $(t,k)$\textsf{-robust} (Theorem~\ref{thm:main-theorem}).
%%%%%%%%%%%%%%%%%%%%%%%%%%%%%%%%%%%%%%%%%
\section{Conclusion \& Future Work}
\label{sec:conclusion}
%%%%%%%%%%%%%%%%%%%%%%%%%%%%%%%%%%%%%%%%%
We analyzed the problem of RC under partially-synchronous network settings. We first relaxed the RC threat model discussed by~\cite{Gramoli2022Trap} by incorporating repeated consensus rounds and discounted utility in each round. In addition, we model types of rational players according to different types $\theta=0,1,2,3$. We showed the impossibility of achieving consensus for $\frac{n}{3} < k + t < \frac{n}{2}$ when players are incentivized to cause either liveness or censorship attacks. In that case, no consensus protocol is \textsf{strongly} $(t,k)$\textsf{robust}. We showed the existence of an insecure Nash Equilibrium in the consensus protocol (TRAP) explained by~\cite{Gramoli2022Trap}. This insecure equilibrium point is preferred over the equilibrium point discussed in TRAP (Section~\ref{ssec:eq-ic}). 

We proposed \proname, a protocol to achieve Ratoinal consensus for $k + t < \frac{n}{2}$ and $t_{0} < \frac{n}{4}$. We proved \textsf{pRFT} security under partially synchronous network and is \textsf{strong} $(t,k)$\textsf{-robust} RC protocol. \proname\ also achieves message complexity $O(n^{3})$ and message size equal to $O(\kappa\cdot n^{4})$. Through this work, we extend the understanding of the RC protocol that does not rely on baiting to achieve consensus. 

\noindent \textbf{Future Work.} We think closing the gap between the impossibilities and guarantees by \textsf{pRFT} is an interesting direction for future work. In addition, improvements on \textsf{pRFT} via use of succinct knowledge f arguments~\cite{Chen2022SNARKS,Ephraim2020SPARK} is left for future work. 

\newpage

\bibliographystyle{unsrtnat}

\newpage
\appendix

\section{Definitions}
\label{app:definitions}
\subsection{Byzantine Agreement}
\begin{definition}[Byzantine Agreement~\cite{blum2021tardigrade} (BA)]
    A protocol $\Pi$ is run by a set of players is $\mathcal{P} = \{\mathcal{P}_{1},\mathcal{P}_{2}\ldots,\mathcal{P}_{n}\}$ and each player $\mathcal{P}_{i}$ initially has value $v_{i}$. This protocol $\Pi$ solves \emph{byzantine agreement} $t$\textsf{-securely} if it satisfies:
    \begin{itemize}[leftmargin=8pt]
        \item \textbf{Agreement:} All altruistic players decide on the same value. 
        \item \textbf{Validity:} If all altruistic players have the same input value $v$ then they agree on the same value $v$.
        \item \textbf{Termination:} All honest players eventually decide on some value and protocol $\Pi$ terminates. 
    \end{itemize}
\end{definition}
Byzantine Agreement (BA) and Byzantine Broadcast (BB) are equivalent problems in the domain of distributed cryptographic protocols. This means that we can use BA (and equivalently BB) as a black box to implement BB (and equivalently BA). BA/BB is possible only if $t < n/3$ in asynchronous and partially-synchronous networks and $t < n/2$ under synchronous network settings~\cite{blum2021tardigrade, LSP1982}.

\subsection{Atomic Broadcast}
\begin{definition}[Atomic Broadcast~\cite{blum2021tardigrade} (ABC)]
    A protocol $\Pi$ which is executed by players $\mathcal{P} = \{\mathcal{P}_{1},\mathcal{P}_{2}\ldots,\mathcal{P}_{n}\}$ who are provided with transactions and maintain a list (chain) of \textsf{Blocks} implements \emph{Atomic Broadcast} $t$\textsf{-securely} if it satisfies the following conditions:
    \begin{itemize}[leftmargin=8pt]
        \item[-] $t$\textsf{-completeness} i.e. if $\leq t$ players are corrupted, then $\forall j > 0$, each honest player outputs block in \emph{iteration} $j$.
        \item[-] $t$\textsf{-consistency} i.e. if $\leq t$ players are corrupted, then $\forall j > 0$ if one honest player outputs \textsf{Block} $B$ in iteration $j$ then all honest players output $B$.
        \item[-] $t$\textsf{-liveness} i.e. if $\leq t$ players are corrupted, then if all honest players have transaction $tx$ as input, then all honest players eventually output a block containing $tx$.
    \end{itemize}
\end{definition}

Atomic Broadcast is the abstraction that is realized by Blockchain protocols. While randomized and probabilistic consensus protocols such as Proof-of-Work (PoW)~\cite{nakamoto2009Bitcoin,garay2015Bitcoin}, Proof-of-Stake (PoS)~\cite{Micali2017Algorand} among others, solve ABC with high probability, protocols such as pBFT~\cite{CastroLiskov1999pBFT}, Honeybadger~\cite{Miller2016Honeybadger} etc. deterministically solve ABC. 

%%%%%%%%%%%%%%%%%%%%%%%%%%%%%%%%%%%%%%%%%
\subsection{Flavours of Synchrony}
\label{ssec:prelims-synchrony}
%%%%%%%%%%%%%%%%%%%%%%%%%%%%%%%%%%%%%%%%%

Consensus and Broadcast problems are studied for a set of players connected through a network. This network is either assumed as \emph{synchronous, asynchronous} or \emph{partially-synchronous}. 

\begin{itemize}
    \item[-] \textbf{Synchronous:} A network is fully-synchronous (or just \emph{synchronous}) if for any message sent from sender $S$ to receiver $R$ reaches within some delay which has an upper bound $\Delta_{sync}$ known to the protocol in advance. Therefore, protocols that work under a synchronous network can be parameterized using this delay parameter $\Delta_{sync}$.
    \item[-] \textbf{Asynchronous:} A network is asynchronous if, for any message from sender $S$ to receiver $R$, the delay has no upper bound but is a finite value. This would mean each message is guaranteed to be delivered, but there exists no upper bound on the delay. 
    \item[-] \textbf{Partially-Synchronous:} Partial-synchronous networks are intermediaries between synchronous networks -- which are difficult to realize and asynchronous  networks -- under which designing protocols is challenging. Partially-Synchronous networks were first discussed by Dwork et al.~\cite{Cynthia1988PartialSynchrony} and are defined as a network having some finite delay $\Delta_{ps}$ set by the adversary and is not known to the protocol. Therefore, a protocol satisfying consensus under partial synchrony cannot be parameterized using $\Delta_{ps}$ and should satisfy any $\Delta_{ps} \in \mathbb{R}_{> 0}$. However, the protocols can use the existence of a finite upper bound to network delay to realize functionalities that were difficult (or impossible) in asynchronous settings. 
\end{itemize}

Consensus through a \emph{deterministic protocol} is impossible under asynchronous network settings in the presence of even a single faulty party~\cite{Fischer1985ImpossibilityAsync}. Consensus under a synchronous network is possible using the deterministic protocol as long as the majority is honest, i.e. byzantine players $< \frac{n}{2}$. Under partial synchrony, consensus is possible if byzantine players are $< \frac{n}{3}$. We discuss RC in Partially-Synchronous network settings and aim to propose a more relaxed adversary model with $t$ byzantine and $k$ rational players and exploit rationality of $k$ players to achieve consensus in $k + t < \frac{n}{2}$.

\section{Proof for Theorem~\ref{thm:r3-impossible}}
\label{app:r3-impossible}

\begin{theorem}[Rational Consensus under $\theta = 3$]
    Under the threat model $\langle(\mathcal{P},T,K),\theta=3,t_{0}\rangle$ no rational consensus protocol is $(t,k)$\textsf{-robust} when $\lceil\frac{n}{3}\rceil \leq k + t \leq \lceil\frac{n}{2}\rceil - 1$.
\end{theorem}
\begin{proof}
    Consider a protocol $\Pi$ belonging to the set of protocols $\mathcal{C}_{3}$ that achieve consensus under threat model $\mathcal{M} = \langle(\mathcal{P},T,K),\theta=3,t_{0}\rangle$. Let $n = |\mathcal{P}|$ therefore, $|T| + |K| < \frac{n}{2}$. In such a case, consider the threshold of messages required for agreement by the protocol to be $\tau$. From Claim~\ref{claim:byz-lim} we have $\tau > \lfloor \frac{n + t_{0}}{2} \rfloor \geq \lceil \frac{n}{2} \rceil$ (since $t_{0} \geq 1$). This means that for consensus, the message/signature of at least one player in $K \cup T$ is required. In this case, if each player $\mathcal{P}_{l} \in K \cup T$ follows strategy $\pi_{abs}$ -- not sending any messages, then consensus is not reached. Since abstaining from sending messages in a round is indistinguishable from message delays due to partially-synchronous networks, $\pi_{abs}$ cannot be distinguished from $\pi_{0}$ under partially-synchronous network. Therefore, $D(\pi_{abs},\sigma) = 0$. The utility for following $\pi_{abs}$ for each rational player is
    \[
        \begin{aligned}
        U_{i}(\pi_{abs},\theta=3) = & \sum_{r=0}^{\infty} \delta^{r}u_{i}(\pi_{abs},3,r) \\ 
        = & E_{\sigma\sim\;S}[f(\sigma,3)] - 0 \\
        = & \alpha > 0 = U_{i}(\pi_{0},\theta=3)
        \end{aligned}
    \]
    Therefore, each player in set $K$ is incentivized to play $\pi_{abs}$ over following protocol which compromises $(t,k)$\textsf{-eventual liveness} property and therefore any such arbitrary $\Pi$ is not a $(t,k)$\textsf{-robust} Rational Consensus protocol.
    \end{proof}

\section{Proof for Theorem~\ref{thm:r2-impossible}}
\label{app:r2-impossible}

\begin{theorem}[Rational Consensus under $\theta=2$]
    Under threat model $\langle(\mathcal{P},T,K),\theta = 2,t_{0}\rangle$ no rational consensus protocol is \textsf{strongly} $(t,k)$\textsf{-robust} when $\lceil\frac{n}{3}\rceil \leq t + k \leq \lceil\frac{n}{2}\rceil - 1$.
\end{theorem}
\begin{proof}
The proof follows by showing a strategy followed by $\mathcal{P}_{i} \in K \cup T$ that is incentivized for rational players and following this strategy, for any protocol $\Pi$ it is impossible to achieve \textsf{strongly} $(t,k)$\textsf{-robust} rational consensus. In Step 1, we describe the strategy, in Step 2, we show that rational players are incentivized to follow this strategy and in Step 3 we show that for any protocol $\Pi$ it is impossible to achieve rational consensus under this strategy.

\smallskip \noindent \underline{Step 1 (Strategy $\pi_{pc}$):} From Theorem~\ref{thm:r3-impossible} we have that for any protocol $\Pi$ with $\tau \geq \lfloor\frac{n + t_{0}}{2}\rfloor + 1$ in any round $r$, the collusion $K\cup T$ can cause disagreement by following $\pi_{abs}$. Let us consider a transaction $tx_{h}$ which is input to all honest players by round $r_{0}$. The strategy which $\mathcal{P}_{i} \in K \cup T$ follows for round $r \geq r_{0}$ is:
\begin{itemize}
    \item[-] If leader in round $r$ is $\mathcal{P}_{l} \not\in K \cup T$ then follow $\pi_{abs}$.
    \item[-] If leader in round $r$ is $\mathcal{P}_{l} \in K \cup T$ then propose \textsf{Block} with transaction set $\overline{tx}$ such that $tx_{h} \not \in \overline{tx}$.
\end{itemize}
We abbreviate this strategy as $\pi_{pc}$ (\textbf{p}artial-\textbf{c}ensorship).

\smallskip \noindent\underline{Step 2 (Incentive Compatibility):} We now show that following $\pi_{pc}$ is incentivized for rational players $\mathcal{P}_{i} \in K$. We first make a simple observation that in from round $r_{0}$ to $r_{0} + n - 1$, in expectation there will be $k + t$ blocks mined (when the leader is $\mathcal{P}_{l} \in K \cup T$). Therefore, the protocol achieves $(t,k)$\textsf{-eventual liveness}. In addition, since there are no conflicting values proposed in any round, no disagreement is reached. The rational players are of type $\theta=2$ which means their utility from round $r_{0}$ onwards is given by 
\[
    U_{i}(\pi,\theta=2) = \sum_{r=r_{0}}^{\infty} \delta^{r}\left(\mathbb{E}[\alpha f(\sigma,2)] - L\cdot D(\pi)\right)
\]
Since system will not be in state $\sigma_{NP}$, the payoff from $\mathbb{E}[\alpha_{2}f_{2}] > 0$ if probability $Pr(\sigma = \sigma_{CP}) > 0$ (according to our strategy $Pr(\sigma = \sigma_{CP}) = 1$). In addition, since there are no duplicate messages signed, and players do not crash forever, following $\pi_{pc}$ is indistinguishable from following $\pi_{0}$ which means $D(\pi_{pc}) = D(\pi_{0}) = 0$. This means $\forall \mathcal{P}_{i} \in K \cup T$
\[
    U_{i}(\pi_{pc},\theta=2) > U_{i}(\pi_{0},\theta=2)
\]
The set $K \cup T$ is therefore incentivized to deviate from the honest protocol $\pi_{0}$ to follow $\pi_{pc}$ for any protocol $\Pi$.

\noindent\underline{Step 3 (No \textsf{strongly} $(t,k)$\textsf{-robust} Rational Consensus:} We now argue that if $K \cup T$ follows $\pi_{pc}$ for any consensus protocol then it is impossible to achieve \textsf{strongly} $(t,k)$\textsf{-robust} rational consensus for any $\frac{n}{3} \leq k + t < \frac{n}{2}$. Consider any round $r \geq r_{0}$. In this round, if leader $\mathcal{P}_{l} \in K \cup T$ the leader selectively includes transactions in transaction set $\overline{tx}$ such that $tx_{h} \not\in \overline{tx}$. If $\mathcal{P}_{l} \not\in K \cup T$ the coalition causes view change without agreement on a block. Thus, any block confirmed (and thus included) doesnot contain transactoin $tx_{h}$ although the transaction is input to all honest players at round $r_{0}$. This violates the $(t,k)$\textsf{-censorship resistance} and therefore for any arbitrary protocol\footnote{we did not assume any property about the protocol} $\Pi$ achieving \textsf{strongly} $(t,k)$\textsf{-robust} rational consensus is impossible.
\end{proof}

\section{Proof for Theorem~\ref{thm:trap-impossible}}
\label{app:trap-impossible}    
\begin{theorem}[Baiting based Rational Consensus under $\theta = 1$]
    Consider any baiting-based rational consensus protocol $\Pi$ the threat model $\mathcal{M} = \langle (\mathcal{P},K,T),\theta=1,t_{0}\rangle$. The set of rational players following $\pi_{fork}$ is a Nash-equilibrium strategy for $|K| > 2 + t_{0} - t$. Under this strategy, the protocol $\Pi$ fails to achieve $(t,k)$\textsf{-robust} rational consensus for $t_{0} = \lceil\frac{n}{3}\rceil - 1$.  
\end{theorem}
\begin{proof}
    To prove this, we show that the payoff for a rational player $\mathcal{P}_{i}$ on following $\pi_{bait}$ is lesser than the payoff if they follow $\pi_{Fork}$ (which is the strategy followed by the collusion). The net gain in payoff for the collusion $K \cup T$ is $G$ if the system ends up in state $\sigma_{Fork}$ which is distributed among the rational players. If the player $\mathcal{P}_{i} \in K$ follows $\pi_{bait}$ then she is not a part of the collusion and will get $0$ payoff if the system ends up in $\sigma_{Fork}$. 

    The payoff for $\mathcal{P}_{l}$ on following $\pi_{Fork}$ along with the rest of the coalition $K \cup T$ is therefore 
    \[
        U_{i}(\pi_{Fork},3) = \frac{G}{k}
    \]
    Consider for $|K| > 2 + t_{0} - t$ i.e. $k + t > 2 + t_{0}$. If $\mathcal{P}_{l}$ deviates from the coalition $K \cup T$ to follow $\pi_{bait}$. The system can end up in two possible states. $\sigma_{Fork}$ is if the rest of the collusion $T \cup K /\{\mathcal{P}_{l}\}$ is able to create a fork/disagreement despite $\mathcal{P}_{l}$ following $\pi_{bait}$. In this case, the utility for $\mathcal{P}_{l}$ is $0$ (since they were not a part of the collusion). Second is if the Proof-of-Fraud submitted by $\mathcal{P}_{l}$ is accepted and the protocol functions normally $\sigma_{0}$. The payoff in this case is 
    \[
        U_{i}(\pi_{bait},3) = R\cdot\;Pr(\sigma=\sigma_{0}) + 0\cdot\;Pr(\sigma=\sigma_{Fork})
    \]
    However, if $k + t > 2 + t_{0}$ then we can achieve disagreement even if $1 + t_{0}$ players follow $\pi_{Fork}$ (since $t_{0} + 1 = \lceil\frac{n}{3}\rceil$). This can be under the partition of $\mathcal{P} / (T \cup K/\{\mathcal{P}_{i}\})$ into two disjoint sets $A,B$ such that $|A|$ and $|B|$ are such that $|A| + k + t \geq \tau$ (and similarly $|B|$). For forking, we have to ensure each partition has at $n - t_{0}$ messages. If $m$ players from the collusion deviates to follow $\pi_{bait}$ the condition for system to \textbf{not} end in $\sigma=\sigma_{Fork}$ is 
    \[
    \begin{aligned}
        |A| + (k - m) + t &< \tau < n - t_{0}\\
        m > |A| + k + t - n + t_{0}
    \end{aligned}
    \]
    We have $n - |B| = |A| + k + t$ and $|B| = \frac{n - t - k}{2}$. Therefore, 
    \[
    \begin{aligned}
        m > t_{0} + \frac{k + t - n}{2}
    \end{aligned}
    \]
For $k > 3 > (3t_{0} + 1) - 2t_{0} - t + 2 = n - 2t_{0} - t + 2$, we have from algebraic reordering the RHS of inequality $> 1$. Thus, any unilateral deviation ($m = 1$) is not sufficient to avoid $\sigma_{Fork}$. Therefore, $Pr(\sigma=\sigma_{Fork}) = 1$ and $Pr(\sigma=\sigma_{0}) = 0$. The utility is therefore $0$ which gives us $U_{i}(\pi_{Fork},3) > U_{i}(\pi_{bait},3)$. Following $\pi_{abs}$ will also lead to $\mathcal{P}_{i}$, not part of $K \cup T$ and therefore the payoff is $0$. Following $\pi_{0}$ leads to payoff $0$ for any system state. Thus, following $\pi_{Fork}$ is Nash Equilibrium strategy for $\mathcal{P}_{l} \in K$.
\end{proof}  

\section{Proof for Claim~\ref{claim:view-change-properties}}
\label{app:view-change-properties}
\begin{claim}
    The \emph{view-change} sub-protocol of \proname\ satisfies both \emph{Consistency} and \emph{Robustness}.
\end{claim}
\begin{proof}
    We prove Consistency through contradiction. Consider $\mathcal{P}_{1} \in A$, $\mathcal{P}_{2} \in B$ such that $A,B \subset H, A \cap B = \emptyset$. Assume consistency does not hold, i.e. $\mathcal{P}_{1}$ commits to view-change and broadcasts $CommitView$ message (ref. Table~\ref{tab:messages}). Since communication channels are reliable and messages are not dropped, the only way to disrupt consistency is if $\mathcal{P}_{2}$ reaches agreement before this $CommitView$ message from $\mathcal{P}_{1}$ reaches $\mathcal{P}_{2}$. For this, we require $|B| + k + t \geq n - t_{0}$. However, for a valid $CommitView$ we require $|A| + k + t \geq n - t_{0}$. Adding them up, $|A| + |B| + 2(k + t) > 2n - 2t_{0}$. This gives us $n + k + t > |A| + |B| + 2(k + t) \geq 2n - 2t_{0}$ and on rearrangement $k + t + 2t_{0} > n$ which is a contradiction for the considered threat model $\mathcal{M} = \langle(\mathcal{P},T,K),\theta=1,\frac{n}{4}\rangle$. Therefore, consistency is satisfied.

    For robustness, consider in a round with honest leader $\mathcal{P}_{l} \in H$. In this case, if $\mathcal{P}_{i} \in T$ broadcast $ViewChange$ message, and abstain from participation, still the protocol is able to gather $\geq n - t_{0}$ message. This is because (1) $t \leq t_{0}$ and (2) rational players are of type $\theta = 1$ due to which they are disincentivized from causing liveness attack. Hence, not enough ($\geq n - t_{0}$) view change messages are gathered and $T$ cannot cause view-change by themselves. The protocol therefore satisfies Robustness property.  
\end{proof}

\section{Proof for Lemma~\ref{lemma:rational-honest}}
\label{app:rational-honest}
\begin{lemma}
    For any $\mathcal{P}_{i} \in K$ under threat model $\mathcal{M} = \langle(\mathcal{P},K,T),\theta=1,\lceil\frac{n}{4}\rceil - 1\rangle$ and protocol \proname, following the protocol honestly (i.e. strategy $\pi_{0}$) is dominant strategy incentive compatible (DSIC) for $|K| + |T| < \frac{n}{2}$ and $t < t_{0}$. 
    \[
        U_{i}(\pi_{0},1) \geq U_{i}(\pi,1) \hspace{10pt}\forall\;\pi,\;\forall\;\mathcal{P}_{i} \in K
    \]
\end{lemma}
\begin{proof}
    Consider any arbitrary rational player $\mathcal{P}_{i} \in K$. We show that by playing $\pi_{fork}$, $\mathcal{P}_{i}$ either: (1) get caught in the PoF, (2) cause view-change or (3) cause agreement on a single value. Consider $\pi_{fork}$ is played by $\mathcal{P}_{i}$ and two honest players $\mathcal{P}_{a}, \mathcal{P}_{b} \in H$ receive conflicting signatures on value $h_{a}, h_{b}$ (such that $h_{a} \not= h_{b}$). First, consider the case when network is synchronous (messages are reaching on time). In this case, signature on $h_{a}$ reaches to $\mathcal{P}_{b}$ and similarly signature on $h_{b}$ reaches $\mathcal{P}_{a}$. Either number of double signatures are $\leq t_{0}$ in which case, agreement is satisfied due to Claim~\ref{claim:agreement-byz}. If number of double signature is $> t_{0}$ then PoD is constructed by either (or both) $\mathcal{P}_{a}$ or (and) $\mathcal{P}_{b}$ and $\mathcal{P}_{i}$ suffers penalty (with some non-zero probability). The payoff in this case is $u_{i}(\pi_{fork},1,r) = - L\cdot D(\pi_{fork},\sigma) < 0$. 
    Consider if the network is partially-synchronous and the network partition of honest nodes is such that $A,B$ are two partitions and $\mathcal{P}_{a} \in A$ and $\mathcal{P}_{b} \in B$. If either partition is small enough that $k + t + |A| < n - t_{0}$ or $k + t + |B| < n - t_{0}$, then agreement is not reached in that partition. In this case, either agreement is not reached for both $\mathcal{P}_{a}$ and $\mathcal{P}_{b}$ or agreement is reached for exactly one of $\mathcal{P}_{a}$ or $\mathcal{P}_{b}$ (therefore on one same value). Let the probability of no agreement be $q_{d}$ and agreement on exactly one value be $q_{a}$. The utility is therefore, $u_{i}(\pi_{fork},1,r) = q_{d}f(\sigma_{NP},1) = -q_{d}\alpha \leq 0$. 
    Notice that it cannot be that agreement is reached in both partitions, as in that case, $|A| + k + t > n - t_{0}$ and $|B| + k + t > n - t_{0}$. As $ |A| + |B| < |H|$, Adding them up, we have $2n - 2t_{0}  \leq n - (k + t) + 2(k + t) \Rightarrow n < k + t + 2t_{0}$.

    However, according to our threat model $\mathcal{M}$, $k + t < \frac{n}{2}$ and $t_{0} < \frac{n}{4}$. Therefore, $k + t + 2t_{0} < n$. Therefore, such a case is not possible in our threat model. Therefore, if the network in (or during) that round is synchronous, $u_{i}(\pi_{fork},1,r) < 0$ and under partially-synchronous network $u_{i}(\pi_{fork},1,r) \leq 0$. Therefore, the expected utility for the round is $\forall i \in [n], r \in \mathbb{R}$ 
    \[
    \begin{aligned}
        \mathbb{E}_{\sigma \sim S}[u_{i}(\pi_{fork},1,r)] \leq 0
        \implies U_{i}(\pi_{fork},1) \leq 0 = U_{i}(\pi_{0},1)
    \end{aligned}
    \]
    Following any other strategy will lead to at most $t_{0}$ double signatures (from byzantine players) and from Claim~\ref{claim:agreement-byz}, the protocol will either reach an agreement or view-change (through timeout or duplicate values proposed by the adversarial leader). Additionally, since rational players are of type $\theta = 1$, they will not try to cause censorship of transactions. Therefore, following $\pi \not\in \{\pi_{0},\pi_{fork}\}$ for $\mathcal{P}_{i} \in K$ gives
    \[
        U_{i}(\pi,1) \leq 0 = U_{i}(\pi_{0},1)
    \]
    Hence, following $\pi_{0}$ gives more payoff than any other strategy. Thus following \proname\ is DSIC for any  $\mathcal{P}_{i} \in K$.
\end{proof}

\section{Construct Proof Procedure}
\label{app:construct-proof}

We elaborate the construct proof procedure invoked in the Reveal phase of the \proname\ protocol (Figure~\ref{fig:protocol}) through Figure~\ref{fig:pof-protocol}.
\begin{figure}[!th]
    \centering
    \begin{mybox}
    {\textsf{ConstructProof}$(M,t_{0})$}
    \begin{algorithmic}[1]
    \STATE $D := \emptyset$
    \FOR{$i \in [n]$, $j \in [n]/\{i\}$}
    
    \FOR{$k \in [n]$}
    \IF{$M(i,k) \not= M(j,k)$}
    \STATE $D \gets D \cup \{(M(i,k),M(j,k)\}$ and \textbf{goto} $9$
    \ENDIF
    \ENDFOR
    \IF{$|D| \geq t_{0} + 1$}
    \STATE \textbf{return} $D$
    \ENDIF
    \ENDFOR
    \STATE \textbf{return} $D$
    \end{algorithmic}
    \end{mybox}
    \caption{Construction of Proof-of-Fraud (PoF)}
    \label{fig:pof-protocol}
\end{figure}


\end{document}